\documentclass[12pt]{amsart}

\usepackage{amsmath,amsthm,amsfonts,amssymb,eucal, bbm}
\usepackage{scalerel}

\usepackage{color}

\renewcommand {\a}{ \alpha }
\renewcommand{\b}{\beta}

\newcommand{\g}{\gamma}

\newcommand{\U}{\Upsilon}

\renewcommand{\d}{\delta}

\renewcommand{\l}{\lambda}

\newcommand{\z}{\zeta}
\renewcommand{\t}{\theta}

\newcommand{\p}{\partial}

\newcommand{\oq}{\ {\raise 7pt\hbox{${\scriptstyle\circ}$}}
	\kern -7pt{
		\hbox{$Q$}}}

\newcommand{\R}{ \mathbb R}

\newcommand {\ba}{\mathbf a}

\newcommand {\BP}{\mathbf P}

\newcommand {\bx}{\mathbf x}

\newcommand {\bk}{\mathbf k}

\newcommand {\bq}{\mathbf q}
\newcommand {\bm}{\mathbf m}
\newcommand {\bl}{\mathbf l}
\newcommand {\bw}{\mathbf w}
\newcommand {\bz}{\mathbf z}
\newcommand {\by}{\mathbf y}

\newcommand {\bn}{\mathbf n}

\newcommand{\SR}{{\sf{R}}}

\newcommand{\SD}{{\sf D}}

\newcommand{\se}{{\sf e}}

\newcommand{\SQ}{{\sf Q}}
\newcommand{\SP}{{\sf P}}
\newcommand{\SfS}{{\sf S}}

\newcommand{\BSP}{\boldsymbol{\sf P}}

\newcommand {\bal}{\boldsymbol\alpha}

\newcommand {\BOLG}{\boldsymbol\Gamma}

\newcommand{\CH}{\mathcal H}

\newcommand{\CM}{\mathcal M}

\newcommand{\CD}{\mathcal D}

\newcommand{\plainW}[1]{\textup{{\textsf{W}}}^{#1}}

\newcommand{\plainC}[1]{\textup{{\textsf{C}}}^{#1}}

\newcommand{\plainL}[1]{\textup{{\textsf{L}}}^{#1}}

\DeclareMathOperator{\card}{{card}}

\DeclareMathOperator {\re} {{Re}}

\DeclareMathOperator{\supp}{{supp}}

\DeclareMathOperator{\dc}{d}

\newtheorem{thm}{Theorem}[section]
\newtheorem{cor}[thm]{Corollary}
\newtheorem{lem}[thm]{Lemma}
\newtheorem{prop}[thm]{Proposition}

\theoremstyle{definition}
\newtheorem{defn}[thm]{Definition}

\newtheorem{rem}[thm]{Remark}

\numberwithin{equation}{section}

\newcommand{\bee}{\begin{equation}}
	\newcommand{\ene}{\end{equation}}
\newcommand{\bees}{\begin{equation*}}
	\newcommand{\enes}{\end{equation*}}
\newcommand{\bes}{\begin{split}}
	\newcommand{\ens}{\end{split}}

\newcommand{\bet}{\begin{thm}}
	\newcommand{\ent}{\end{thm}}
\newcommand{\bel}{\begin{lem}}
	\newcommand{\enl}{\end{lem}}
\newcommand{\bec}{\begin{cor}}
	\newcommand{\enc}{\end{cor}}
\newcommand{\bep}{\begin{proof}}
	\newcommand{\enp}{\end{proof}}
\newcommand{\ber}{\begin{rem}}
	\newcommand{\enr}{\end{rem}}

\newcommand{\1}{\mathbbm 1}

\newcommand{\scalel}[1]
{{\scaleto{#1}{3pt}}}

\begin{document}
	\hoffset -4pc

\title
[One-particle density matrix]
{The diagonal behaviour of the one-particle Coulombic density matrix}
\author{Peter Hearnshaw}
\author{Alexander V. Sobolev}
\address{Department of Mathematics\\ University College London\\
	Gower Street\\ London\\ WC1E 6BT UK}
\email{peter.hearnshaw.18@ucl.ac.uk}
\email{a.sobolev@ucl.ac.uk}

\begin{abstract}	 
 We obtain bounds for all derivatives of the non-relativistic 
 Coulombic one-particle density matrix $\g(x, y)$ near the diagonal $x = y$. 
\end{abstract}


\maketitle

\section{Introduction}

Consider on $\plainL2(\R^{3N})$ the Schr\"odinger operator 
\begin{align}
\CH = &\ \CH_0 + V,\quad \CH_0 = - \Delta = - \sum_{k=1}^N \Delta_k,\notag\\
V(\bx) = &\ - Z
\sum_{k=1}^N \frac{1}{|x_k|}  
 + \sum_{1\le j< k\le N} \frac{1}{|x_j-x_k|},\label{eq:potential}
\end{align}
describing an atom with $N$ particles 
(e.g. electrons)  
with coordinates $\bx = (x_1, x_2, \dots, x_N)$, $x_k\in\R^3$, $k= 1, 2, \dots, N$, 
and a nucleus with charge $Z>0$.  
The notation $\Delta_k$ is used for 
the Laplacian w.r.t. the variable $x_k$. 
The operator $\CH$ acts on the Hilbert space $\plainL2(\R^{3N})$ and by
 standard methods one proves that it is self-adjoint on the domain 
$D(\CH) =\plainW{2, 2}(\R^{3N})$, 
see e.g. \cite[Theorem X.16]{ReedSimon2}. (Here and throughout the paper we use the standard notation 
$\plainW{l, p}$ for the Sobolev spaces, where $l$ and $p$ indicate the smoothness and summability 
respectively). 
Our methods allow consideration of the molecular 
Schr\"odinger operator, but we restrict our attention to the atomic case for simplicity. 
Let $\psi = \psi(\bx)$,  
be an eigenfunction of the operator $\CH$ with an eigenvalue $E\in\R$, i.e. $\psi\in D(\CH)$ and 
\begin{align*}
(\CH-E)\psi = 0.
\end{align*}
For each $j=1, \dots, N$, we represent
\begin{align*}
\bx = (x_j, \hat\bx_j), \quad \textup{where}\ 
\hat\bx_j = (x_1, \dots, x_{j-1}, x_{j+1},\dots, x_N),
\end{align*}
with obvious modifications if $j=1$ or $j=N$. 
The one-particle density matrix is defined as the function 
\begin{align}\label{eq:den}
\tilde\g(x, y) = \sum_{j=1}^N\int\limits_{\R^{3N-3}}\psi(x, \hat\bx_j) 
 \overline{\psi(y, \hat\bx_j)}\,  d\hat\bx_j,\quad (x,y)\in\R^3\times\R^3. 
\end{align} 
This function is one of the key objects in 
the  multi-particle quantum mechanics, see 
\cite{RDM2000}, \cite{Davidson1976}, \cite{LLS2019}, \cite{LiebSei2010} for details and futher references. 
If one assumes that all $N$ particles are spinless 
fermions (resp. bosons), i.e. that the function $\psi$ is 
antisymmetric (resp. symmetric) under the permutations $x_j\leftrightarrow x_k$, 
then the definition \eqref{eq:den} simplifies:
\begin{align*}
\tilde\g(x, y) = N \int_{\R^{3N-3}} \psi(x, \hat\bx)  \overline{\psi(y, \hat\bx)} d\hat\bx,\ 
\quad \textup{where} \ \hat\bx = \hat\bx_1. 
\end{align*}
Our main result however does not require any symmetry assumptions. 
We are interested in the smoothness properties of the function \eqref{eq:den}. It is clear that 
for this purpose it suffices to study each term in \eqref{eq:den} individually. 
Moreover, using permutations of the variables it is sufficient to focus just on one term on the 
right-hand side of \eqref{eq:den}:
\begin{align}\label{eq:gamma}
\g(x, y) = \int_{\R^{3N-3}}
\psi(x, \hat\bx)  \overline{\psi(y, \hat\bx)} \, d\hat\bx,\ 
\quad \hat\bx = (x_2, x_3, \dots, x_N).
\end{align}
Throughout the paper we refer to this function as the \textit{one-particle density matrix}. 
In \cite{FHOS2004} the \textit{one-particle density} 
\begin{align}\label{eq:dens}
\rho(x) = \g(x, x) = \int_{\R^{3N-3}}
|\psi(x, \hat\bx)|^2 \, d\hat\bx,
\end{align}
was shown to be 
a real-analytic function of $x\not = 0$. 
The real analyticity of the function $\g(x, y)$ as a 
function of two variables on the domain
\begin{align}\label{eq:D}
\CD = \{(x, y): |x|\,|y|\not = 0,\quad x\not = y\}\subset \R^3\times\R^3,
\end{align} 
%
%
%
was proved in \cite{HearnSob2022}.  
As was pointed out in \cite{HearnSob2022},  
one cannot expect analyticity in $x$ and $y$ to hold on the diagonal $x = y$. 
In fact, quantum chemistry calculations in 
\cite{Cioslowski2020} (see also \cite{Cioslowski2021}) show that in the neighbourhood of the diagonal 
$x=y\not = 0$, the function $\g(x, y)$ has the following behaviour:
\begin{align}\label{eq:cio}
\re\g(x, y) - \g(x, x)= C(x) |x-y|^5 + o(|x-y|^5),\quad \textup{as}\quad y\to x,
\end{align}
with some non-zero function $C(x)$. 
Motivated in part by the formula \eqref{eq:cio}, 
in the current paper we aim to obtain 
explicit bounds for all partial derivatives of $\g(x, y)$  
near the diagonal $x = y$ or the points $x = 0$ and $y=0$.
For the derivatives of $\g$ we use the standard notation $\p_x^m \p_y^l\g(x, y)$, 
where $m, l\in\mathbb N_0^3$, 
$\mathbb N_0 = \mathbb N\cup\{0\}$.  
%
%
For two non-negative numbers (or functions) 
$X$ and $Y$ depending on some parameters, 
we write $X\lesssim Y$ (or $Y\gtrsim X$) if $X\le C Y$ with 
some positive constant $C$ independent of those parameters. 
The notation $B(x, R)$ is used for the open ball of radius $R$ centred at the point $x$.  

%
%

For $b\ge 0$, $t > 0$ we define 
\begin{align}\label{eq:h}
 h_b(t) = 
\begin{cases} 
1,\ \textup{if}\ b <5,\\
\log\big(t^{-1}+2),\ \textup{if}\ b=5,\\
1 + t^{5-b},\ \textup{if}\ b >5.
\end{cases}
\end{align}
The next theorem contains the main result.  

\begin{thm}\label{thm:derbounds} 
Let $\psi$ be an eigenfunction, and let $R>0$ be fixed.  
Then for all $x \not = 0, y\not = 0$, $x\not =y$, and all $l, m\in\mathbb N_0^3$ such that 
$|l|\ge 1, |m|\ge 1$, we have
\begin{align}\label{eq:der2}
|\p_x^l\p_y^m \g(x, y)|\lesssim   \big(1 + |x|^{2-|l|-|m|} 
+ |y|^{2-|l|-|m|} 
 & +  h_{|l|+|m|}(|x-y|)\big)\notag\\
  &\ \times \big(\|\rho\|_{\plainL1(B(x, R))}\big)^{\frac{1}{2}}
\big(\|\rho\|_{\plainL1(B(y, R))}\big)^{\frac{1}{2}}.
\end{align}
Furthermore, for all $|l|\ge 1$,
\begin{align}\label{eq:der1}
|\p_x^l\g(x, y)| + |\p_y^l\g(x, y)|\lesssim \big(1 + |x|^{1 - |l|} 
+ |y|^{1 - |l|}  
&\ + h_{|l|}(|x-y|)\big)\notag\\ 
&\ \times \big(\|\rho\|_{\plainL1(B(x, R))}\big)^{\frac{1}{2}}
\big(\|\rho\|_{\plainL1(B(y, R))}\big)^{\frac{1}{2}}.
\end{align}
The implicit constants in \eqref{eq:der2} and \eqref{eq:der1} 
may depend on $l$, $m$ and $R$, but are independent of $x$, $y$ and the function $\psi$.
\end{thm}

\begin{rem}\label{rem:main}
\begin{enumerate}
\item 
Theorem \ref{thm:derbounds} naturally extends  
the case of a molecule with several nuclei whose positions
are fixed. The modifications are straightforward.
\item 
Observe that $\|\rho\|_{\plainL1(B(x, R))}\le \|\psi\|_{\plainL2(\R^{3N})}^2$, $R>0$,  
so that the right-hand sides of 
\eqref{eq:der2} and \eqref{eq:der1} are finite. 
\item 
The bound \eqref{eq:der1} 
for $|l| = 1$ implies that $\g(x, y)$ is a Lipschitz function on $\R^3\times\R^3$. 
\item 
%
The bounds \eqref{eq:der2}, \eqref{eq:der1} ensure that 
$\g\in\plainW{5, p}_{\textup{\tiny loc}}(\CD)$ 
with arbitrary $p <\infty$, see \eqref{eq:D} for the definition of $\CD$. 
By Proposition \ref{prop:remove} in the Appendix, this implies that 
$\g\in\plainW{5, p}_{\textup{\tiny loc}}\big((\R^3\setminus\{0\})\times(\R^3\setminus\{0\})\big)$, 
which means 
that $\g\in\plainC{4, \t}_{\textup{\tiny loc}}
\big((\R^3\setminus\{0\})\times(\R^3\setminus\{0\})\big)$ 
for all $\t <1$. This result just barely misses the $\plainC{4, 1}$-smoothness 
of the factor $|x-y|^5$ in the formula \eqref{eq:cio}. 
In this sense Theorem 
\ref{thm:derbounds} is sharp, up to a $\log$-term.  
\item 
The one-electron density $\rho(x) = \g(x, x)$ 
is known to be real analytic for $x\not = 0$, see \cite{FHOS2004}. 
As shown in \cite{FS2021}, the function $\rho(x)$ satisfies the bound 
\begin{align*}
|\p_x^l\rho(x)|\lesssim \big(1 + |x|^{1 - |l|}\big)\, \|\rho\|_{\plainL1(B(x, R))},
\quad \textup{for all} \quad l\in\mathbb N_0^3, \quad x\not = 0.
\end{align*}
%
For $|l|\le 4$ this bound 
%
%
follows from \eqref{eq:der2} and 
\eqref{eq:der1}.
\item  
There is an independent (indirect) 
argument indicating that $\g(x, y)$ should have 
a $|x-y|^5$-singularity on the diagonal.  More precisely, 
we can show that if the function $\g(x, y)$ has on the diagonal 
a singularity of the type $|x-y|^b$ with some $b >-3$, then necessarily $b = 5$. 
This argument is based on the analysis of spectral asymptotics 
of the (non-negative) operator 
$\BOLG$ with kernel $\tilde\g(x, y)$. It was shown in \cite{Sobolev2022}
that the eigenvalues $\l_k(\BOLG)$ have asymptotics of order $k^{-8/3}$ as $k\to\infty$. 
%
%
On the other hand, according to the results 
on spectral asymptotics for integral operators 
with homogeneous kernels, see \cite{BS1970} (and also 
\cite{Sobolev2022} for a summary), 
the singularity $|x-y|^b$, $b\not = 0, 2, 4, \dots$,  
would produce the spectral asymptotics of order $k^{-(1+b/3)}$. 
The exponent $1+b/3$ coincides with $8/3$ exactly  for $b=5$, which proves the point.  
\end{enumerate}
\end{rem}

Our proofs have two main ingredients. 
At the heart of our method are regularity properties of the eigenfunction $\psi$. 
By the standard elliptic argument, the function $\psi$ is real analytic away from the singularities 
of the potential \eqref{eq:potential}, i.e. away from the particle coalescence points. According   
to T. Kato's seminal paper \cite{Kato1957}, at the coalescence points the function 
$\psi$ is Lipschitz. A detailed study of smoothness properties of $\psi$ 
was conducted in the recent paper 
\cite{FS2021}, to which we also refer for further bibliography. In particular, this paper provides global 
pointwise bounds for partial derivatives of $\psi$. 
In the study of the one-particle density \eqref{eq:dens} conducted in \cite{FS2021} the key point was 
to obtain bounds for certain directional derivatives of $\psi$. 
Such derivatives are also critical for our analysis in the current paper, and we explain their importance 
below.  

Let $\SR = \{1, 2, \dots, N\}$ be the set of all particle labels.  
A subset $\SP\subset\SR$ is called \textit{cluster}. 
%
%
For each cluster $\SP$ we define the following 
\textit{cluster (directional) derivative}: 
\begin{align}\label{eq:clusterder}
\SD_{\SP}^m  = &\ \biggl(\sum_{k\in \SP} \p_{x_k'}\biggr)^{m'}
\bigg(\sum_{k\in \SP} \p_{x_k''}\bigg)^{m''}
\bigg(\sum_{k\in \SP} \p_{x_k'''} \bigg)^{m'''}, 
\end{align}
Here $m = (m', m'', m''')$ with $m', m'', m'''\in \mathbb N_0 := \mathbb N\cup \{0\}$ and 
$x = (x', x'', x''')$ with $x', x'', x'''\in\R$. Observe that for all $m\not = 0$, 
\begin{align*}
\SD_\SP^m \frac{1}{|x_j-x_k|} = 0,\quad \textup{if}\quad j, k\in \SP \quad\textup{or}
\quad j, k\notin \SP.
\end{align*}   
Therefore the potential \eqref{eq:potential} is infinitely smooth with respect to $\SD_\SP$ as long as 
$x_j\not = 0, j = 1, 2, \dots, N,$ and $x_j\not = x_k$, where $j\in \SP, k\notin \SP$ 
or $j\notin \SP, k\in \SP$. 
As a consequence, the function $\psi$ is also infinitely smooth with respect to $\SD_\SP$ 
for the same values of the coordinates, 
see \cite{FHOS2004} or \cite{HearnSob2022}. In particular, the cluster derivatives of $\psi$ 
do not have singularities at the coalescence points $x_j = x_k$, if $j, k\in \SP$ 
or $j, k\notin \SP$. 
One of the pivotal points in \cite{FS2021} was the pointwise 
bound for the cluster derivative $\SD_\SP^m\psi(\bx)$ with explicit dependence on the distance 
of $\bx\in\R^{3N}$ to the coalescence set
\begin{align*}
\Sigma_{\SP} = \bigg\{\bx\in\R^{3N}: \prod_{j\in\SP} |x_j| 
\prod_{k\in\SP, l\in \SP^{\rm c}}|x_k-x_l| = 0\bigg\}.
\end{align*}
For our purposes we need bounds of such type 
%
%
 %
 %
 for cluster derivatives 
involving an arbitrary finite number $M$ of clusters $\SP_1, \SP_2, \dots, \SP_M$,  
i.e. for $\SD_{\SP_1}^{m_1}\SD_{\SP_2}^{m_2}\cdots\SD_{\SP_M}^{m_M}\psi$, 
see Corollary \ref{cor:clusterpsi}. 
This generalization is not immediate and requires substantial further work, which is done in Sect. 
\ref{sect:reg} and  \ref{sect:schr}.
%
%
%
%

The next step of the proof is to use the bounds obtained for cluster 
derivatives to estimate partial derivatives of the function \eqref{eq:gamma}. To 
this end we use a partition of unity consisting of smooth functions $\Phi(x, y, \hat\bx)$ of 
$3N+3$ variables that we call \textit{extended cut-off functions}, and study the integrals 
\begin{align}\label{eq:cutdens}
\g(x, y; \Phi) = \int_{\R^{3N-3}} \Phi(x, y, \hat\bx)
\psi(x, \hat\bx) \overline{\psi(y, \hat\bx)} \, d\hat\bx.
\end{align}
Extended cut-offs were introduced in 
\cite{HearnSob2022} to establish the 
real analyticity of $\g(x, y)$ outside the diagonal. 

The support of each extended cut-off divides the particles $x_2, x_3, \dots, x_N$ 
into three disjoint groups: the first two 
groups consist of particles that are ``close" to the particles $x$ and $y$ respectively, and the third one 
contains the particles that are ``far" from the first two groups. 
%
%
%
%
This partition naturally gives rise to two clusters, denoted $\SP$ and $\SfS$: cluster $\SP$ 
labels the particles close 
to $x$, and cluster $\SfS$ -- the particles close to $y$. 
%
%
Thus, when differentiating the integral 
\eqref{eq:cutdens} w.r.t. $x$ and $y$, under the integral 
these derivatives convert into the cluster derivatives $\SD_\SP$ and $\SD_\SfS$ respectively.  
The actual calculation is more involved, but we can 
illustrate this conversion using the following simplified example.    
Let us differentiate with respect to $x$ the integral  
\begin{align*}
F(x) = \int f(x, \hat\bx) d\hat\bx, 
\end{align*}
where the integration is conducted over the space $\R^{3N-3}$, and we assume for simplicity that $f\in \plainC\infty_0(\R^{3N})$. 
To this end under the integral we make the change of variables 
$\hat\bx = \hat\bw + \hat\bz$, where 
$\hat\bz = (z_2, z_3, \dots, z_N)\in\R^{3N-3}$ is defined by 
$z_j = x, j\in \SP,$ and $z_j = 0, j\notin \SP$, for some cluster $\SP$. 
Thus  $F(x)$ rewrites as 
\begin{align*}
F(x) = \int f(x, \hat\bw + \hat\bz)\, d\hat\bw.
\end{align*}
Consequently, for all $l\in\mathbb N_0^3$, we have  
\begin{align*}
\p_x^{l}F(x) = \int \big(\SD_\SP^l f\big)(x, \hat\bw+\hat\bz)\,   d\hat\bw
 = \int \SD_\SP^l f(x, \hat\bx)\, d\hat\bx.
\end{align*}  
In order to estimate the derivatives of the integral 
\eqref{eq:cutdens}, we use the bounds for cluster derivatives of $\psi$ 
obtained in the first stage of the proof. Integrating these bounds in
$\hat\bx$ leads to \eqref{eq:der2} and \eqref{eq:der1} thereby completing the proof. 
One should say that the application these bounds is not immediate, 
but we defer the discussion of this technical step 
until Section \ref{sect:estim}, see Remark \ref{rem:grp}.

The paper is organized as follows. 
In Sect. \ref{sect:clusters} we gather information about 
cut-off functions and clusters associated with them. Most of the required facts are borrowed from
\cite{HearnSob2022}. Sect. \ref{sect:reg} 
considers the general elliptic equation of the form \eqref{eq:mod} in a ball 
$B(\bx_0, \ell) = \{\bx: |\bx-\bx_0|< \ell\}\subset \R^{3N}$. 
Here we study cluster derivatives of solutions of \eqref{eq:mod}, and the main focus is on the 
explicit dependence  of these estimates on the radius $\ell$, see Theorem \ref{thm:reg2}. 
In Sect. \ref{sect:schr} these estimates are applied to the Schr\"odinger equation to derive the bounds 
for the cluster derivatives of $\psi$, summarized in Corollary \ref{cor:clusterpsi}. 
Some estimates for integrals 
emerging in the proof of Theorem \ref{thm:derbounds} are gathered in Section \ref{sect:aux}. 
The proof itself 
%
%
is completed in Sect. \ref{sect:estim}. 
The appendix (Section \ref{sect:app}) 
contains an elementary extension property for the Sobolev spaces that was proved in \cite{Sobolev2022a}.  
 
\textbf{Notation.} 
We conclude the introduction with some general notational conventions.  

%
%

\textit{Coordinates.} 
As mentioned earlier, we use the following standard notation for the coordinates: 
$\bx = (x_1, x_2, \dots, x_N)$,\ where $x_j\in \R^3$, $j = 1, 2, \dots, N$. 
As a rule we represent $\bx$ in the form 
$\bx = (x_1, \hat\bx)$ with  
$\hat\bx = (x_2, x_3, \dots, x_N)\in\R^{3N-3}$. 

For $N\ge 3$ it is also useful to introduce the notation 
for $\hat\bx$ with $x_j$, $j\ge 2$, taken out. Let 
\begin{align}\label{eq:xtilde}
\tilde\bx_j = (x_2, \dots, x_{j-1}, x_{j+1},\dots, x_N),  
\end{align}
so that $\hat\bx = (x_j, \tilde\bx_j)$ and $\bx = (x_1, x_j, \tilde\bx_j)$.

\textit{Clusters.}
Let $\SR = \{1, 2, \dots, N\}$. 
A subset $\SP\subset\SR$ is called 
\textit{cluster}. 
%
%
We denote $|\SP| = \card \SP$, 
$\SP^{\rm c} = \SR\setminus \SP$,\ $\SP^* = \SP\setminus\{1\}$. 
If $\SP = \varnothing$, then $|\SP| = 0$ and $\SP^{\rm c} = \SR$.

For $M$ clusters $\SP_1, \dots, \SP_M$ we write $\BP = \{\SP_1, \SP_2, \dots, \SP_M\}$, 
%
%
and call $\BP$  \textit{cluster set}. 
Clusters $\SP_1, \SP_2, \dots, \SP_M$ in a cluster set are not assumed to be all disjoint or distinct. 

\textit{Derivatives.} 
Let $\mathbb N_0 = \mathbb N\cup\{0\}$.
If $x = (x', x'', x''')\in \R^3$ and $m = (m', m'', m''')\in \mathbb N_0^3$, then 
the derivative $\p_x^m$ is defined in the standard way:
\begin{align*}
\p_x^m = \p_{x'}^{m'}\p_{x''}^{m''}\p_{x'''}^{m'''}.
\end{align*}
This notation extends to $x\in\R^d$ with an arbitrary dimension $d\ge 1$ in the obvious way. 
Denote also 
\begin{align*}
\p^{\bm} =  
\p_{x_1}^{m_1} \p_{x_2}^{m_2}\cdots \p_{x_N}^{m_N},\quad 
\bm = (m_1, m_2, \dots, m_N)\in \mathbb N_0^{3N}.
\end{align*}
%
%
%
%
Let $\BP = \{\SP_1, \SP_2, \dots, \SP_M\}$ be a cluster set, and let 
$\bm = (m_1, m_2, \dots, m_M)$, $m_k\in \mathbb N_0^3$, $k = 1, 2, \dots, M$. Then we denote 
\begin{align*} 
\SD_{\BP}^{\bm} = \SD_{\SP_1}^{m_1} \SD_{\SP_2}^{m_2}\cdots\SD_{\SP_M}^{m_M},
\end{align*} 
where each individual cluster derivative is defined as in \eqref{eq:clusterder}.  
It is easy to see that the cluster derivatives satisfy the Leibniz rule. 
We use this fact without further comments throughout the paper.  

\textit{Supports.} 
For any smooth function $f = f(\bx)$,  
we define $\supp_0 f = \{\bx: f(\bx)\not = 0\}$. It is clear that the closure $\overline{\supp_0 f}$ 
coincides with the support $\supp f$ defined in the standard way. 
With this definition we immediately get the useful 
property that
\begin{align*}
\supp_0 (fg) = \supp_0 f\cap \supp_0 g.
\end{align*}

\textit{Bounds.} 
As explained earlier, for two non-negative numbers (or functions) 
$X$ and $Y$ depending on some parameters, 
we write $X\lesssim Y$ (or $Y\gtrsim X$) if $X\le C Y$ with 
some positive constant $C$ independent of those parameters. 
If $X\lesssim Y$ and $Y\lesssim X$, then $X\asymp Y$. 
To avoid confusion we often make explicit comments on the nature of 
(implicit) constants in the bounds. In particular, 
all constants (implicit or explicit) may depend on the eigenvalue $E$, 
the number of particles $N$ and the charge $Z$.

\section{Cut-off functions and clusters}\label{sect:clusters}

\subsection{Admissible cut-off functions}  
Let 
\begin{align}\label{eq:xi}
\xi\in \plainC\infty_0(\R):\ 0\le \xi(t)\le 1,\ 
\xi(t) = 
\begin{cases}
1,\ {\rm if}\ |t|\le 1,\\[0.2cm]
0,\ {\rm if}\ |t|\ge 2.
\end{cases}
\end{align}
Now for $\varepsilon>0$ we define two  radially-symmetric functions $\z\in \plainC\infty_0(\R^3)$, 
$\t\in\plainC\infty(\R^3)$ as follows:
\begin{align}\label{eq:pu}
\z(x) = \z_\varepsilon(x) 
= \xi\bigg(\frac{4N}{\varepsilon}|x|\bigg),\quad 
\t(x) = \t_\varepsilon(x) = 1-\z_\varepsilon(x), \quad x\in\R^3,
\end{align}
so that 
\begin{align*}
\z(x) = 0 \quad \textup{for}\quad x\notin B\big(0, \varepsilon(2N)^{-1}\big),\qquad\qquad
\t(x) = 0 \quad \textup{for}\quad x\in B\big(0, \varepsilon(4N)^{-1}\big). 
\end{align*}
The dependence of the cut-offs on the parameter $\varepsilon$ is important, 
but it is not always reflected in the notation.   

Our next step is to build out of the functions $\z_\varepsilon$ and $\t_\varepsilon$ 
cut-off functions of $3N$ variables. 
Let $\{f_{jk}\}, 1\le j, k\le N$, be a set of functions such that 
each of them is one of the functions $\z_\varepsilon$ 
or $\t_\varepsilon$, and 
$f_{jk} = f_{kj}$.
We call functions of the form 
\begin{align}\label{eq:canon}
\phi(\bx) = \prod\limits_{1\le j < k\le N} f_{j k}(x_j-x_k).
\end{align} 
\textit{admissible cut-off functions}
or simply \textit{admissible cut-offs}. 
Such cut-offs (or, more precisely, a slightly more general version thereof) were used 
in \cite{FHOS2002}, \cite{FHOS2004} and also in 
\cite{HearnSob2022}. 
We need only a subset of their properties established 
in \cite{FHOS2004} and \cite{HearnSob2022}. 
  
We associate with the function $\phi$ a cluster $\SQ(\phi)$ defined next. 

\begin{defn} 
For an admissible cut-off $\phi$, let $I(\phi)\subset \{(j, k)\in\SR\times\SR: j\not= k\}$
be the index set such that 
$(j, k)\in I(\phi)$, iff $f_{jk} = \z$.    
We say that two indices $j, k\in\SR$, 
are $\phi$-\textit{linked} to each other 
if either $j=k$, or $(j, k)\in I(\phi)$, or 
there exists a sequence of pairwise 
distinct indices $j_1, j_2, \dots, j_s$, $1\le s\le N-2$, 
all distinct from $j$ and $k$, 
such that $(j, j_1), (j_s, k)\in I(\phi)$ and $(j_p, j_{p+1})\in I(\phi)$ for all $p = 1, 2, \dots, s-1$.

The cluster $\SQ(\phi)$ associated with the cut-off $\phi$ 
is defined as the set of all indices that are $\phi$-linked to index $1$.  
\end{defn}

It follows from the above definition that $\SQ(\phi)$ always contains index $1$. Note also that 
the notion of being linked defines an equivalence relation on $\SR$, and the cluster $\SQ(\phi)$ is 
nothing but the equivalence class of index $1$. 
On the support of the admissible cut-off $\phi$ the variables 
$x_j$, indexed by $j\in\SQ(\phi)$, are ``close" to each other and are  
``far" from the remaining variables. In order to quantify these facts below 
we define a number of subsets in $\R^{3N}$ and $\R^{3N-3}$. 

For any cluster $\SP$ we 
introduce the following sets depending on the parameter $\varepsilon>0$:
\begin{align*}
X_{\SP}(\varepsilon) = 
\begin{cases}
\R^{3N}\quad {\rm for}\ |\SP| = 0 \ \textup{or} \ N,\\[0.2cm]
\{\bx\in \R^{3N}: \ |x_j-x_k| > \varepsilon, 
\forall j\in \SP, k\in \SP^{\rm c} \},\quad  {\rm for}\ 0< |\SP| < N,
\end{cases}
\end{align*} 
The set $X_\SP(\varepsilon)$ separates the points $x_k$ 
and $x_j$ labeled by the clusters $\SP$ and $\SP^{\rm c}$ respectively. 
Note that $X_{\SP}(\varepsilon) = X_{\SP^{\rm c}}(\varepsilon)$. 
Define also the sets separating $x_k$'s from the origin:
\begin{align*}
T_\SP(\varepsilon) = 
\begin{cases}
\R^{3N}, \quad {\rm for}\ |\SP|=0,\\[0.2cm]
\{\bx\in\R^{3N}: |x_j|>\varepsilon,\ \forall j\in \SP\},\ \quad {\rm for}\ |\SP|>0.
\end{cases}
\end{align*}
It is also convenient to introduce 
corresponding sets in the space $\R^{3N-3}$:
%
%
\begin{align}\label{eq:hatxp}
\widehat X_{\SP}(x, \varepsilon) = 
\{\hat\bx\in \R^{3N-3}: (x, \hat\bx)\in X_{\SP}(\varepsilon)\},\quad  \textup{for all}\quad 
x\in\R^3, 
\end{align} 
and 
\begin{align}\label{eq:hattp}
\widehat T_{\SP}(\varepsilon) = 
\begin{cases}
\R^{3N-3}, \quad {\rm for}\ |\SP^*|=0,\\[0.2cm]
\{\hat\bx\in\R^{3N-3}: |x_j|>\varepsilon,\ \forall j\in \SP^*\},\ \quad {\rm for}\ |\SP^*|>0.
\end{cases}
\end{align}
Observe that $\widehat T_{\SP}(\varepsilon) = \widehat T_{\SP^*}(\varepsilon)$.

The support of the admissible cut-off $\phi$ is easily described with the help of the sets introduced above.  
The next proposition is adapted from 
\cite[Lemma 4.3(i)]{FHOS2004} and 
\cite[Lemmata 4.2, 4.3]{HearnSob2022}. 

\begin{prop}\label{prop:supp} 
For $\SP = \SQ(\phi)$ the inclusion 
\begin{align*}
\supp_0\phi
\subset X_{\SP}\big(\varepsilon(4N)^{-1}\big)
\end{align*}
holds. 
 
Moreover, if $j\in \SQ(\phi)$, then $|x_1-x_j|<\varepsilon/2$ for all $\bx\in\supp_0\phi$. 
If $|x_1|> \varepsilon$, then 
\begin{align*}
\supp_0\phi(x_1, \ \cdot\ ) 
\subset \widehat T_{\SP}(\varepsilon/2).
\end{align*}
\end{prop}
 We do not use Proposition \ref{prop:supp} directly in this paper, but 
present it in order to demonstrate the relevance of the associated cluster $\SQ(\phi)$.

\subsection{Extended cut-offs} \label{subsect:extended} 
In our analysis the central role 
is played by another class of cut-off functions.   
These cut-offs are  
functions of $3N+3$ variables and they are defined as follows: for each  
$x, y\in\R^3, \hat\bx\in\R^{3N-3}$ let 
\begin{align}\label{eq:Phi}
\Phi(x, y, \hat\bx) = \prod_{2\le j\le N} g_{j}(x-x_j)
\prod_{2\le j\le N} h_{j}(y-x_j)
\prod_{2 \le k < l\le N} f_{kl}(x_k-x_l),
\end{align} 
 where each of the functions $g_j, h_j$ and $f_{jk} = f_{kj}$  is one of the cut-offs 
 $\t$ or $\z$ defined in \eqref{eq:pu}. We call such functions 
 \textit{extended} cut-offs. 
Each extended cut-off uniquely defines two admissible cut-offs:
\begin{align*}
\phi(x, \hat\bx) = &\ \prod_{2\le j\le N} g_{j}(x-x_j)
\prod_{2 \le k < l\le N} f_{kl}(x_k-x_l),\\
\mu(y, \hat\bx) = &\ 
\prod_{2\le j\le N} h_{j}(y-x_j)
\prod_{2 \le k < l\le N} f_{kl}(x_k-x_l),
\end{align*}
see definition \eqref{eq:canon}. 
We say that the pair $\phi$, $\mu$ and the extended 
cut-off $\Phi$ are associated to each other. 
We denote by $\SP = \SQ(\phi)$ and $\SfS = \SQ(\mu)$ 
the clusters associated with $\phi$ and $\mu$ respectively. 

We remind that the functions $\phi, \mu$ and $\Phi$ all
depend on the parameter $\varepsilon$. Thus whenever necessary 
we include $\varepsilon$ in the notation and write, for example, $\Phi(x, y, \hat\bx; \varepsilon)$. 
Note however that the clusters $\SP$ and $\SfS$ do not depend on $\varepsilon$.

 Below we list some useful 
properties 
of the extended cut-offs $\Phi$ and associated admissible $\phi$, $\mu$ adapted from 
\cite[Lemmata 4.6, 4.8]{HearnSob2022}.  

\begin{prop}\label{prop:empty}
If $\SP^*\cap \SfS$ is non-empty 
and $|x-y| > \varepsilon$, 
then $\Phi(x, y, \hat\bx; \varepsilon) = 0$ 
for all $\hat\bx\in\R^{3N-3}$. 
\end{prop}

In the estimates further on we assume, as a rule, that $|x-y|>\varepsilon$, so due 
to Proposition \ref{prop:empty} from now on we may suppose that 
$\SP^*\subset\SfS^{\rm c}$ (which is equivalent to $\SfS^*\subset\SP^{\rm c}$). 
Under this condition we obtain 
the following information about the support of the extended cut-off $\Phi$:
%
%

\begin{prop} 
If $\SP^*\subset \SfS^{\rm c}$, then for all $x, y\in \R^3$ we have 
\begin{align}\label{eq:phimux}
\supp_0 \Phi(x, y, \ \cdot\ ; \varepsilon)
\subset \widehat X_{\SP}\big(x, \varepsilon(4N)^{-1}\big)
\cap & \widehat X_{\SfS^*}\big(x, \varepsilon(4N)^{-1}\big)\notag\\
& \cap \widehat X_{\SfS}\big(y, \varepsilon(4N)^{-1}\big)\cap \widehat X_{\SP^*}\big(y, \varepsilon(4N)^{-1}\big).
\end{align}
If, in addition $|x|>\varepsilon, |y|>\varepsilon$, then 
\begin{align}\label{eq:phihat}
\supp_0\Phi(x, y, \ \cdot\ ; \varepsilon) \subset \widehat 
T_{\SP^*}(\varepsilon/2)\cap \widehat T_{\SfS^*}(\varepsilon/2).
\end{align}
\end{prop}  
 
To complete this section we need to make a remark on the cluster derivatives of the 
extended cut-offs. For any cluster $\SQ\subset\SR$ 
we denote by $\SD_{x, \SQ}^m\Phi(x, y, \hat\bx)$ 
the cluster derivative of $\Phi$ as a function of the variables $(x, \hat\bx)$, i.e. 
\begin{align}\label{eq:clustermul}
\SD_{x,\SQ}^m
%
%
\Phi(x, y, \hat\bx) = \SD_{\SQ}^m\tilde \Phi_y(x, \hat\bx),\quad \textup{where}\quad 
\tilde \Phi_y(x, \hat\bx) := \Phi(x, y, \hat\bx).    
 \end{align}
Similarly we define the derivative $\SD_{y,\SQ}^m\Phi(x, y, \hat\bx)$. 
It immediately follows from the definition of $\Phi(x, y, \hat\bx)$ that 
for any clusters $\SQ_1, \SQ_2$ and all $\bm = (m_1, m_2)\in\mathbb N_0^6$,
the bound holds:
\begin{align}\label{eq:phib}
|\SD_{x,\SQ_1}^{m_1} \SD_{y,\SQ_2}^{m_2}\Phi(x, y, \hat\bx; \varepsilon)| 
\lesssim 
\begin{cases}
1, \quad \textup{if}\quad |\bm|=0,\\[0.2cm]
\varepsilon^{-|\bm|} \CM_\varepsilon(x, y, \hat\bx),\quad \textup{if}\quad |\bm|\ge 1,
\end{cases}
\end{align}
where  
\begin{align}\label{eq:phibm}
\CM_\varepsilon(x, y, \hat\bx) = &\ \sum_{2\le j\le N} 
\xi\big(N\varepsilon^{-1}|x-x_j|\big)\notag \\
&\ + \sum_{2\le j\le N} \xi\big(N\varepsilon^{-1}|y-x_j|\big)
+ \sum_{2\le j < k\le N} \xi\big(N\varepsilon^{-1}|x_j-x_k|\big),
\end{align}
see \eqref{eq:xi} for the definition of the function $\xi$.

\section{Regularity estimates}\label{sect:reg}

\subsection{$\plainC1$-regularity for elliptic equations}
In what follows we rely on the well-known $\plainC1$-regularity bounds 
for solutions of second order elliptic equations on bounded domains.  
This type of regularity is discussed e.g. in \cite[Ch. 3]{LadUra1968} and  
\cite[Ch. 8]{GilTru2001}. To be precise, in \cite[Ch. 3]{LadUra1968} and  
\cite[Ch. 8]{GilTru2001} one can find bounds even 
in the space  $\plainC{1,\t}$ with a suitable $\t\in (0, 1)$, but  
we are not concerned with the H\"older properties.  
Moreover, in this paper we do not need the most general form of the equation.  
For our purposes it suffices to consider the equation 
\begin{align}\label{eq:mod}
\big(-\Delta + \ba(x)\cdot\nabla + b(x)\big) u = g, 
\end{align}
on an open ball $B\subset \R^d$, where 
all the coefficients are $\plainL\infty(B)$-functions. 
The proposition below provides some convenient $\max$-bounds for the 
weak solution and its first derivatives. Since the proof is quite short, we 
provide it for the sake of completeness. 
Note that the proof of \cite[Proposition A.2]{FS2021} 
uses similar ideas.  

\begin{prop}\label{prop:reg} 
Let $B_R = B(x_0, R)\subset \R^d$ for some $x_0\in\R^d$ and $R>0$. 
Suppose that $u\in \plainW{1, 2}(B_R)$ is a 
weak solution of the equation \eqref{eq:mod}, 
where $\ba, b, g\in\plainL\infty(B_R)$, and 
\begin{align*}
\|\ba\|_{\plainL\infty(B_R)} + \|b\|_{\plainL\infty(B_R)}
\le M, 
\end{align*}
with some constant $M>0$. 
Then for any $r \in (0, R)$ the function $u$ belongs 
to $\plainW{2, 2}(B_r)\cap\plainC1(\overline{B_r})$ and 
\begin{align}\label{eq:regc1}
\|u\|_{\plainC1(\overline{B_r})}\lesssim \|u\|_{\plainL2(B_R)}
+ \|g\|_{\plainL\infty(B_R)},
\end{align}
with an implicit constant that depends only on the constant 
$M$, dimension $d$ and the radii $r$ and $R$. 
\end{prop} 

\begin{proof} 
The inclusion $u\in \plainW{2, 2}(B_r)$ is a direct consequence of the interior regularity. 

In order to prove that the weak solution $u\in \plainW{1, 2}(B_R)$ has 
the $\plainC1$-regularity in $B_r$ we repeatedly apply the following 
elementary fact.

Assume that $u\in \plainW{1, p}(B_\rho)$ with some $p \in( 1, \infty)$ and $\rho \le R$. Then for any 
$\nu < \rho$  the following is true:
\begin{enumerate}
\item 
If $p\le d$, then $u\in \plainW{1, q}(B_\nu)$ with $q = p(1+d^{-1})$ and 
\begin{align}\label{eq:step}
\|u\|_{\plainW{1,q}(B_\nu)}\lesssim &\ \|u\|_{\plainL{p}(B_\rho)} + \|g\|_{\plainL{p}(B_\rho)}\notag\\
\lesssim &\ \|u\|_{\plainW{1,p}(B_\rho)} + \|g\|_{\plainL{\infty}(B_\rho)}. 
\end{align}
\item 
If $p > d$, then $u\in \plainC1(\overline{B_\nu})$ and 
\begin{align}\label{eq:regc11}
\|u\|_{\plainC1(\overline{B_\nu})}\lesssim &\ \|u\|_{\plainL{p}(B_\rho)} 
+ \|g\|_{\plainL{p}(B_\rho)}\notag\\
\lesssim &\ \|u\|_{\plainW{1,p}(B_\rho)} + \|g\|_{\plainL{\infty}(B_\rho)}.
\end{align}
\end{enumerate}

Indeed, under the assumption $u\in \plainW{1, p}(B_\rho)$, 
by the interior regularity we have $u\in\plainW{2, p}(B_\nu)$ and the standard bound 
holds (see, e.g. \cite[Theorem 9.11]{GilTru2001}):
\begin{align}\label{eq:gt}
\|u\|_{\plainW{2,p}(B_\nu)}\lesssim \|u\|_{\plainL{p}(B_\rho)} + \|g\|_{\plainL{p}(B_\rho)}.
\end{align}
If $p < d$, then we use the bounded embedding $\plainW{2,p}\subset \plainW{1, q}$,\ for all 
$q\in[p, p^*]$, $p^* = dp(d-p)^{-1}$. In particular, the value 
$q = p(1+d^{-1})$ belongs to the interval $[p, p^*]$, which proves  
\eqref{eq:step}. 
If $p = d$, then $\plainW{2,p}\subset \plainW{1, q}$ for all $q\in [p, \infty)$, and hence 
for $q = p(1+d^{-1})$ in particular. Hence \eqref{eq:step} holds again.

If $p >d$, then we use the embedding $\plainW{2,p}\subset \plainC{1}$, so that 
\eqref{eq:gt} leads to \eqref{eq:regc11}.

Let us proceed with the proof of \eqref{eq:regc1}. 
If $d = 1$, then the solution $u\in\plainW{1, 2}(R)$ satisfies \eqref{eq:regc11} with 
$p=2$, which immediately implies \eqref{eq:regc1}.
Suppose that $d\ge 2$, and define the sequence 
\begin{align*}
q_n = 2(1+d^{-1})^n, \ n = 0, 1, \dots. 
\end{align*}
Let $k\ge 1$ be the index such that $q_{k-1}\le d$ and $q_{k}>d$.  
Pick finitely many numbers $r_n>0$ such that 
\begin{align*}
r < r_k<r_{k-1}<\dots < r_1 < r_0  = R.
\end{align*} 
Since $u\in\plainW{1,2}(B_R)$ we can apply the bound \eqref{eq:step} with $\nu = r_1, \rho = R$ 
and $p = q_0 = 2$, $q = q_1$. Repeating this step successively for $\rho =r_n, \nu =r_{n+1}$ and 
$p = q_{n}$, $q = q_{n+1}$  
for all $n = 1, \dots, k$, we arrive at the bound 
\begin{align*}
\|u\|_{\plainW{1,q_{k}}(B_{r_k})} 
\lesssim &\ 
\|u\|_{\plainW{1, q_{k-1}}(B_{r_{k-1}})} + \|g\|_{\plainL{\infty}(B_{r_{k-1}})}
\lesssim \dots\\
\lesssim  &\ \|u\|_{\plainW{1, q_{\scalel{1}}}(B_{r_{\scalel{1}}})} 
+ \|g\|_{\plainL{\infty}(B_{r_{\scalel{1}}})}
\lesssim 
\|u\|_{\plainL{2}(B_R)} + \|g\|_{\plainL{\infty}(B_R)}.
\end{align*}
As $q_k > d$, we can now use \eqref{eq:regc11} which gives
%
\begin{align*}
\|u\|_{\plainC1(\overline{B_r})} 
\lesssim &\ 
\|u\|_{\plainW{1,q_{k}}(B_{r_k})} +  \|g\|_{\plainL{\infty}(B_{r_k})}\\
\lesssim &\ 
\|u\|_{\plainL{2}(B_R)} + \|g\|_{\plainL{\infty}(B_R)}.
 \end{align*}
This completes the proof of \eqref{eq:regc1}. 
\end{proof}

\subsection{Cluster derivatives}
The next result is tailored for later use with the multi-particle Schr\"odinger equation. 
We assume that in the equation \eqref{eq:mod} $d = 3N$ and 
the variable $\bx$ is given by $\bx = (x_1, x_2, \dots, x_N)$.  

Now we obtain better regularity properties of the weak solution of \eqref{eq:mod} 
assuming some additional smoothness of the coefficients with respect to cluster 
derivatives. 
Precisely, consider a weak solution $u$ of the equation \eqref{eq:mod} in the ball 
$B(\bx_0, R\ell)$ with some $\bx_0\in\R^{3N}$, $R >0$, $\ell\in (0, 1]$. 
Suppose that for a cluster set $\BSP = \{\SP_1, \SP_2, \dots, \SP_M\}$ 
%
%
the coefficients $\ba$ and $b$ 
%
%
satisfy the bounds 
\begin{align}\label{eq:cd}
|\SD_{\BSP}^{\bm}\ba(\bx)| + 
|\SD_{\BSP}^{\bm}b(\bx)|\lesssim \ell^{-|\bm|}, \quad \bx\in B(\bx_0, R\ell),
\end{align}
for all $\bm\in\mathbb N_0^{3M}$, with 
constants potentially depending on $\bm$, $R$ and $\bx_0$, but not on $\ell$. 
In the next theorem we obtain bounds for the cluster derivatives $\SD_{\BSP}^{\bm} u$ 
with explicit dependence on the parameter $\ell\in (0, 1]$.

\begin{thm}\label{thm:reg2} 
%
Assume the conditions \eqref{eq:cd}, and let $u$ be a weak solution of the equation 
\eqref{eq:mod} in $B(\bx_0, R\ell)$ 
 with the right-hand side $g = 0$.  
Then for all $\bm\in\mathbb N_0^{3M}$ and all $r<R$ 
the cluster derivatives $\SD_{\BSP}^{\bm} u$ belong to 
$\plainC1\big(\overline{B(\bx_0, r\ell)}\big)$. 
Furthermore, if $|\bm|+k\ge 1$, where $k = 0, 1$, then for all $\nu\in (r, R)$ we have 
\begin{align}\label{eq:reg2}
\|\nabla^k \SD_{\BSP}^{\bm} u\|_{\plainL\infty(B(\bx_0, r\ell))}
\lesssim \ell^{1-|\bm|-k}\, \big(\ell\|u\|_{\plainL\infty(B(\bx_0, \nu\ell))}
+ \|\nabla u\|_{\plainL\infty(B(\bx_0, \nu\ell))}\big).
\end{align}
If $|\bm|\ge 2$, then also  
\begin{align}\label{eq:d2ell}
\|\SD_{\BSP}^{\bm} & u\|_{\plainL\infty(B(\bx_0, r\ell))}\notag\\
\lesssim &\ \ell^{2-|\bm|}\,\big(
\max_{\bl: |\bl|=2} \|\SD_{\BSP}^{\bl} u\|_{\plainL\infty(B(\bx_0, \nu \ell))}
+ \|\nabla u\|_{\plainL\infty(B(\bx_0, \nu\ell))}
+  \|u\|_{\plainL\infty(B(\bx_0, \nu\ell))}\big).
\end{align}
The implicit constants in \eqref{eq:reg2} and \eqref{eq:d2ell} may depend 
%
%
on the constants $r, \nu, R$, order $\bm$ of the derivative, cluster set $\BSP$, and the 
constants in \eqref{eq:cd}. In particular, 
if the constants in \eqref{eq:cd} are independent of $\bx_0$, then 
so are the constants in \eqref{eq:reg2} and \eqref{eq:d2ell}.
\end{thm}

In the present paper we need the bound \eqref{eq:reg2} only.  
Nevertheless, the bound \eqref{eq:d2ell} may be of independent interest.
 
The idea of the proof follows a similar argument in \cite{FS2021}. Namely, we rescale the 
problem by introducing the function 
\begin{align*}
w(\by) = u(\bx_0 + \ell \by),\quad \by\in B(\bold0, R),
\end{align*}
and the coefficients 
\begin{align*}
\bal(\by) =  \ba(\bx_0+\ell \by),\quad 
\b(\by) = b(\bx_0+\ell \by).
\end{align*}
After the scaling the equation \eqref{eq:mod} takes the form
\begin{align}\label{eq:mod_scale}
-\Delta w + 2\ell\ \bal\cdot\nabla w
+  \ell^2 \beta w = 0.
\end{align}
Note that 
\begin{align}\label{eq:cd_scale}
|\SD_{\BSP}^{\bm}\bal(\by)| + 
|\SD_{\BSP}^{\bm}\b(\by)|\lesssim 1, \quad \by\in B(\bold0, R),
\end{align}
for all $\bm\in\mathbb N_0^{3M}$.

We are interested in the bounds for the function 
$w_\bm = \SD_{\BSP}^{\bm} w$ for $\by\in B(\bold0, R)$. 
For the sake of brevity throughout the proof we use the notation 
$B_r = B(\bold0, r)$ for $r >0$.

\begin{lem}
Let $\ell\le 1$.  Suppose that $w$ is a weak solution of the equation \eqref{eq:mod_scale}. 
Then for all $\bm\in\mathbb N_0^{3M}$ and all $r<R$ 
the cluster derivatives $w_\bm$ belong to 
$\plainC1\big(\overline{B_r})\big)$. 
Furthermore, if $|\bm|\ge 1$, then for all $\nu\in (r, R)$ we have  
\begin{align}\label{eq:reg_scale}
\|w_\bm\|_{\plainC1(\overline{B_{r}})}
\lesssim \ell^2\, \|w\|_{\plainL\infty(B_{\nu})} + \|\nabla w\|_{\plainL\infty(B_{\nu})}.
\end{align}
If $|\bm|\ge 2$, then also 
\begin{align}\label{eq:d2}
\|w_\bm\|_{\plainC1(\overline{B_r})}
\lesssim 
\max_{\bl: |\bl|=2} \|\SD_{\BSP}^{\bl} w\|_{\plainL\infty(B_\nu)}
+ \ell \|\nabla w\|_{\plainL\infty(B_\nu)}
+ \ell^2 \|w\|_{\plainL\infty(B_\nu)}.
\end{align}
\end{lem} 

\begin{proof} 
We begin the proof with a formal manipulation assuming 
that all the cluster derivatives $w_\bm$ exist and are as smooth as necessary. 
Applying the operator $\SD_{\BSP}^{\bm}$ to the equation \eqref{eq:mod_scale} 
and using \eqref{eq:cd_scale} we obtain the following equation for the function 
$w_\bm$:
\begin{align}\label{eq:cluster}
-\Delta w_\bm + 2\ell\ \bal\cdot\nabla w_\bm  + \ell^2 \b\, w_\bm = g_\bm,
\end{align}
with 
\begin{align*}
g_\bm = -2 \ell \sum\limits_
{\substack{\bold0\le \bq\le \bm\\
|\bq|\le |\bm|-1}} {\bm\choose \bq}\big(\SD_{\BSP}^{\bm-\bq} \bal\big)\cdot \nabla w_{\bq}
-\ell^2 
\sum\limits_
{\substack{\bold0\le \bq\le \bm\\
|\bq|\le |\bm|-1}} 
{\bm\choose \bq}\big(\SD_{\BSP}^{\bm-\bq} \b\big)\, w_{\bq}.
\end{align*}
%
%
%
%
It follows from \eqref{eq:cd_scale} that 
\begin{align}\label{eq:prq1}
\|g_\bm\|_{\plainL\infty(B_\rho)}
\lesssim  \sum_{\bq:0\le |\bq|\le |\bm|-1} T_{\bq}(\ell, \rho), \quad 
T_{\bq}(\ell, \rho) = \ell\|\nabla w_\bq\|_{\plainL\infty(B_\rho)} 
+ \ell^2 \|w_\bq\|_{\plainL\infty(B_\rho)},
\end{align}
for all $\rho < R$. 

Now we need to justify the above formal calculations. First note that 
since $w$ is a weak solution of \eqref{eq:mod_scale}, by Proposition \ref{prop:reg} it belongs to 
$\plainW{2, 2}(B_\rho)\cap\plainC1(\overline{B_\rho})$ for all $\rho <R$, so that 
$T_\bold0(\ell, \rho) = 
\ell\|\nabla w\|_{\plainL\infty(B_\rho)} + \ell^2 \|w\|_{\plainL\infty(B_\rho)}$ is finite, 
and as a result, for $|\bm| = 1$ we have 
\begin{align}\label{eq:prq}
\|g_{\bm}\|_{\plainL\infty(B_\rho)}\lesssim 
T_\bold0(\ell, \rho) = 
\ell\|\nabla w\|_{\plainL\infty(B_{\rho})} + \ell^2\, \|w\|_{\plainL\infty(B_{\rho})},\quad 
0<\rho <R.
\end{align} 
By Proposition \ref{prop:reg}, for all $\bm$ such that $|\bm|=1$ and all $r < \rho$, we conclude that 
$w_{\bm} \in \plainW{2, 2}(B_r)\cap\plainC1(\overline{B_r})$ and 
\begin{align*}
\|w_\bm\|_{\plainC1(\overline{B_{r}})}
\lesssim &\ \|w_\bm\|_{\plainL\infty(B_\rho)} + \|g_\bm\|_{\plainL\infty(B_\rho)}\\
\lesssim &\ \|\nabla w\|_{\plainL\infty(B_\rho)} + \|g_\bm\|_{\plainL\infty(B_\rho)}
\lesssim  \|\nabla w\|_{\plainL\infty(B_{\rho})} + \ell^2\, \|w\|_{\plainL\infty(B_{\rho})}.
\end{align*}
This proves \eqref{eq:reg_scale} for $|\bm|=1$.  

Using the above observation as the induction base, 
we now prove 
%
%
that $w_\bm$ is indeed well-defined for all $\bm\in\mathbb N_0^{3M}$
 and that it 
is a weak solution of 
the equation \eqref{eq:cluster}. 
%
%
To provide the induction step 
assume that for some $k = 1, \dots,$ and all $\bm, 1\le |\bm| \le k$, 
the function $w_\bm$ is a weak solution 
of \eqref{eq:cluster} in $B_\rho$ 
for all $\rho < R$ (and hence belongs to 
$\plainW{2, 2}(B_r)\cap\plainC1(\overline{B_r}), r < \rho,$ by Proposition \ref{prop:reg}), 
and that it satisfies 
\eqref{eq:reg_scale}. 
Let us prove that the same is true for
the function $w_\bn$ for all $\bn$ such that $|\bn|=k+1$.

First note that for all $\rho <\nu < R$ 
the function $g_{\bn}$ satisfies the bound 
\begin{align}\label{eq:prq2}
\|g_{\bn}\|_{\plainL\infty(B_\rho)}\lesssim 
\ell\|\nabla w\|_{\plainL\infty(B_{\nu})} + \ell^2\, \|w\|_{\plainL\infty(B_{\nu})}.
\end{align}  
Indeed, in the bound 
\eqref{eq:prq1} 
for $T_{\bold0}(\ell, \rho)$ we use the estimate \eqref{eq:prq}.  
For $T_\bq(\ell, \rho)$ with $1\le |\bq|\le k$ we use 
\eqref{eq:reg_scale} to obtain
the estimate
\begin{align*}
T_\bq(\ell, \rho) 
\lesssim  \ell \|w_\bq\|_{\plainC1(\overline{B_{\rho}})} 
%
%
\lesssim \ell\,\|\nabla w\|_{\plainL\infty(B_{\nu})} 
+ \ell^3\, \|w\|_{\plainL\infty(B_{\nu})}. 
\end{align*}
In view of \eqref{eq:prq1}, this  
gives \eqref{eq:prq2} for $g_\bn$, as required. 
Furthermore, as $w_\bm\in \plainW{2, 2}(B_\rho)$ for all 
$|\bm|\le k$, we have  
$w_{\bn}\in \plainW{1, 2}(B_\rho)$. Now,  
 integrating the equation \eqref{eq:mod_scale} 
 against the function $(-1)^{k+1}\SD_{\BSP}^{\bn} \eta$, 
 with an arbitrary $\eta\in\plainC\infty_0(B_\rho)$ we obtain 
 \begin{align*}
0 = (-1)^{k+1}\int \nabla w\cdot \nabla\ \SD_{\BSP}^{\bn} \eta\, d\by
 + (-1)^{k+1} 2\ell &\ \int \bal \cdot \nabla w \,\SD_{\BSP}^{\bn} \eta\, d\by\\ 
 + &\ (-1)^{k+1} \ell^2\int \b w\, \SD_{\BSP}^{\bn} \eta\, d\by
  \end{align*}
 \begin{align*}
= \int \nabla w_\bn\cdot \nabla\ \eta\, d\by
 +   2\ell &\ \int \bal\cdot \nabla w_\bn \, \eta\, d\by\\ 
 + &\ \ell^2\int \b w_\bn  \eta\, d\by 
 -  \int g_{\bn}\, \eta\, d\by, 
 \end{align*}
and hence $w_{\bn}$ is a weak solution of \eqref{eq:cluster} in $B_\rho$. Since the coefficients 
and the right-hand side 
of \eqref{eq:cluster} 
are bounded uniformly in $\ell$, by  Proposition \ref{prop:reg}, 
$w_\bn\in\plainW{2, 2}(B_r)\cap\plainC1(\overline{B_r})$ for all $r<\rho$ and 
\begin{align}\label{eq:n2}
\|w_{\bn}\|_{\plainC{1}(\overline{B_r})}
\lesssim &\ \|w_{\bn}\|_{\plainL\infty(B_\rho)}
+ \|g_{\bn}\|_{\plainL\infty(B_\rho)}\notag\\
\lesssim &\ \|w_{\bn}\|_{\plainL\infty(B_\rho)}
+ \ell\|\nabla w\|_{\plainL\infty(B_\nu)}
+ \ell^2\|w\|_{\plainL\infty(B_\nu)}, 
\end{align}
where we have also used \eqref{eq:prq2}. 
Let $\bm, |\bm|=k$, be such that $|\bn-\bm|=1$. Therefore, by \eqref{eq:reg_scale},  
\begin{align*}
\|w_\bn\|_{\plainL\infty(B_\rho)}\lesssim \|\nabla w_\bm\|_{\plainL\infty(B_\rho)}
\lesssim \|\nabla w\|_{\plainL\infty(B_\nu)} + \ell^2\, \|w\|_{\plainL\infty(B_\nu)}.
\end{align*}
Together with \eqref{eq:n2} this leads to 
\begin{align*}
\|w_{\bn}\|_{\plainC{1}(\overline{B_r})}
\lesssim \|\nabla w\|_{\plainL\infty(B_\nu)}
+ \ell^2\|w\|_{\plainL\infty(B_\nu)},
\end{align*}
as required. 
It remains to conclude that by induction, the bound \eqref{eq:reg_scale} 
holds for all $\bm\in\mathbb N_0^{3M}$.  
  
Proof of \eqref{eq:d2} is also conducted by induction.  
 Assume that $|\bm|=2$. From \eqref{eq:n2} we get that 
\begin{align}\label{eq:m2}
\|w_{\bm}\|_{\plainC{1}(\overline{B_r})}
\lesssim &\ 
\|w_{\bm}\|_{\plainL\infty(B_\rho)} + \ell\|\nabla w\|_{\plainL\infty(B_\nu)}
+ \ell^2\|w\|_{\plainL\infty(B_\nu)} \notag\\
\lesssim &\ 
\max_{\bl: |\bl|=2} \|\SD_{\BSP}^{\bl} w\|_{\plainL\infty(B_\nu)} + \ell \|\nabla w\|_{\plainL\infty(B_\nu)}
+ \ell^2 \|w\|_{\plainL2(B_\nu)}, 
 \end{align}
for all $r < \rho < \nu<R$, which gives \eqref{eq:d2}. 
The bound \eqref{eq:m2} serves as the induction base. Let us now provide the induction step. 
Suppose that \eqref{eq:d2} holds for all $\bm$ such that $2\le |\bm|\le k$ 
with some $k = 2, 3, \dots$. Let us prove that it holds for all $w_{\bn}$ where $|\bn|=k+1$. 
Let $\bm, |\bm|=k$, be such that $|\bn-\bm|=1$. 
Thus \eqref{eq:d2} for $w_\bm$ implies that  
\begin{align*}
\|w_\bn\|_{\plainL\infty(B_\rho)}\lesssim \|\nabla w_\bm\|_{\plainL\infty(B_\rho)} 
\lesssim 
\max_{\bl: |\bl|=2} \|\SD_{\BSP}^{\bl} w\|_{\plainL\infty(B_\nu)} + \ell \|\nabla w\|_{\plainL\infty(B_\nu)}
+ \ell^2 \|w\|_{\plainL2(B_\nu)}.
\end{align*}
Substituting this inequality in \eqref{eq:n2}, we obtain \eqref{eq:d2} for $w_\bn$. 

Consequently, \eqref{eq:d2} holds for all $\bm\in \mathbb N_0^{3M}$, $|\bm|\ge 2$. 
\end{proof}

We would like to point out one fact which was not needed in the above proof but deserves mentioning. 
The bound \eqref{eq:reg_scale} for $|\bm|=1$ and bound \eqref{eq:d2} for $|\bm|\ge 2$ 
imply \eqref{eq:reg_scale} for all $|\bm|\ge 2$. Indeed, 
by \eqref{eq:reg_scale} with $|\bq|=1$, we have  
\begin{align*}
 \max_{\bl: |\bl|=2} \|\SD_{\BSP}^{\bl} w\|_{\plainL\infty(B_\rho)}\lesssim \sum_{|\bq|=1}\|\nabla w_\bq\|_{\plainL\infty(B_\rho)}
 \lesssim \|\nabla w\|_{\plainL\infty(B_\nu)}
+ \ell^2 \|w\|_{\plainL\infty(B_\nu)},
\end{align*} 
for all $\rho < \nu<  R$. 
After substitution in \eqref{eq:d2} this gives \eqref{eq:reg_scale} for all $|\bm|\ge 2$, 
as claimed. 

\begin{proof}[Proof of Theorem \ref{thm:reg2}]
Since
\begin{align*}
 \nabla^k w_{\bm}(\by) = \ell^{|\bm|+k}\big(\nabla^k 
 \SD_{\BSP}^{\bm} u\big)(\bx_0+\ell \by), \ k = 0, 1,\quad \bm\in\mathbb N_0^{3M},
\end{align*}
the bound \eqref{eq:reg_scale} for $|\bm|\ge 1$ rewrites as \eqref{eq:reg2}. If $\bm = \bold0$ 
and $|k| = 1$, then \eqref{eq:reg2} is trivial. 

The bound \eqref{eq:d2ell} is obtained from \eqref{eq:d2} in the same way. 
\end{proof}

\section{The Schr\"odinger equation}\label{sect:schr}

\subsection{Reduction of the Schr\"odinger equation}
Here we apply the bounds obtained in the previous section to the Schr\"odinger equation
\begin{align}\label{eq:H}
%
%
-\Delta\psi + (V-E)\psi = 0,
\end{align}
where the potential $V$ is defined in \eqref{eq:potential}. 
Throughout this section we do not impose the condition $\psi\in\plainL2(\R^{3N})$, 
but consider local solutions of the equation \eqref{eq:H}. 
In order to reduce 
\eqref{eq:H} to the equation of the form \eqref{eq:mod} we use the  
representation $\psi = e^F \phi$ with a 
function $F$ such that 
\begin{align}\label{eq:Finf}
F, \nabla F\in \plainL\infty(\R^{3N}).
\end{align} 
In the mathematical physics literature the function $e^F$ is often called 
the \textit{Jastrow factor}, see e.g. \cite{FS2021}, \cite{HaKlKoTe2012}. 
After the substitution the equation \eqref{eq:H} rewrites as 
\begin{align}\label{eq:jastrow}
-\Delta \phi - 2\nabla F\cdot \nabla\phi + (V- \Delta F - |\nabla F|^2-E)\phi = 0. 
\end{align}
More precisely, if $\psi\in \plainW{1, 2}_{\tiny {\rm loc}}(\R^{3N})$ is a weak 
solution of \eqref{eq:H}, then we also have $\phi\in \plainW{1, 2}_{\tiny {\rm loc}}(\R^{3N})$ 
and $\phi$ is a weak solution of \eqref{eq:jastrow}. 
We make a standard choice of  the function $F$ which ensures that 
%
%
the factor 
$\phi$ has better smoothness properties than $\psi$: 
\begin{align}\label{eq:F}
F(\bx) = \sum_{j=1}^N \bigg(-\frac{Z}{2}|x_j| 
+ &\ \frac{1}{4}\sum_{j<k\le N}|x_j-x_k|\bigg)  \notag\\
+ &\ \sum_{k=1}^N\bigg(
\frac{Z}{2}\sqrt{|x_j|^2+1}
- \frac{1}{4}\sum_{j < k\le N}
\sqrt{|x_j-x_k|^2 + 1}
\bigg). 
\end{align} 
In fact, the conventional choice of $F$ is given by the first sum 
on the right-hand side, which has a bounded gradient, 
whereas the second term is added to ensure that $F$ itself is bounded on $\R^{3N}$.  
Thus both conditions \eqref{eq:Finf} are satisfied. 
Due to the straightforward identities 
\begin{align*}
\Delta |x_j-x_k| = \frac{4}{|x_j-x_k|},\quad \Delta |x_j| = \frac{2}{|x_j|},
\end{align*} 
we have 
\begin{align*}
\Delta\bigg[\sum_{j=1}^N \bigg(-\frac{Z}{2}|x_j| 
+ \frac{1}{4}\sum_{j<k\le N}|x_j-x_k|\bigg)\bigg] = V(\bx).
\end{align*}
Therefore 
\begin{align}\label{eq:z}
\nabla^k(-\Delta F +  V)\in \plainL\infty(\R^{3N}),\ \quad \textup{for all}
\quad k = 0, 1, \dots.
\end{align}
The following fact was proved in 
\cite[Proposition A.2]{FS2021}. We provide a somewhat different (shorter) 
proof that uses only Proposition 
\ref{prop:reg}. 

\begin{lem} \label{lem:grpsi}
Let $\psi\in \plainW{1, 2}_{\tiny {\rm loc}}(\R^{3N})$ be a weak solution of \eqref{eq:H}. 
Then $\psi \in \plainW{1,\infty}_{\tiny {\rm loc}}(\R^{3N})$ and for all numbers 
$r>0$ and $R>0$ such that $r < R$ we have  
\begin{align}\label{eq:grpsi}
\|\psi\|_{\plainL\infty(B(\bx_0, r))} + \|\nabla\psi\|_{\plainL\infty(B(\bx_0, r))} 
\lesssim \|\psi\|_{\plainL2(B(\bx_0, R))},
\end{align}
where the implicit constant does not depend on $\psi$ and $\bx_0\in \R^{3N}$, 
but depends on $r$ and $R$. 
\end{lem}

\begin{proof} 
The function $\phi = e^{-F}\psi$ is a weak solution of the equation \eqref{eq:jastrow}  
in $B(\bx_0, R)$ with arbitrary $\bx_0\in\R^{3N}$ and $R>0$. As established above, the coefficients 
in this equation are uniformly bounded on $\R^{3N}$. According to 
Proposition \ref{prop:reg}, $\phi\in \plainC1(\R^{3N})$ and for any pair of radii 
$r, R$, $r <R$, we have the bound 
\begin{align}\label{eq:grphi}
\|\phi\|_{\plainL\infty(B_r)} 
+ \|\nabla\phi\|_{\plainL\infty(B_r)}\lesssim \|\phi\|_{\plainL2(B_R)},
\end{align}
where $B_R = B(\bx_0, R)$, and the implicit constant is independent of $\bx_0$. 
Since $F$ and $\nabla F$ satisfy \eqref{eq:Finf}, we can write for $\psi$ that 
\begin{align*}
\|\psi\|_{\plainL\infty(B_r)}\lesssim\|\phi\|_{\plainL\infty(B_r)}, 
\quad 
\|\phi&\|_{\plainL2(B_R)}\lesssim\|\psi\|_{\plainL2(B_R)} \quad \textup{and}\\[0.2cm]
\|\nabla\psi\|_{\plainL\infty(B_r)}
\lesssim \|\nabla\phi\|_{\plainL\infty(B_r)} & + \|\phi\|_{\plainL\infty(B_r)}. 
\end{align*}
Therefore \eqref{eq:grphi} implies \eqref{eq:grpsi}, as required. 
\end{proof}

\subsection{Cluster derivatives of $\psi$: application of Theorem \ref{thm:reg2}}
For a cluster $\SP\subset \{1, 2, \dots, N\}$ introduce the set
\begin{align*}
\Sigma_{\SP} = \bigg\{\bx\in\R^{3N}: \prod_{j\in\SP} |x_j| 
\prod_{k\in\SP, l\in \SP^{\rm c}}|x_k-x_l| = 0\bigg\},
\end{align*}
and the distance
\begin{align}\label{eq:dq}
\dc_{\SP}(\bx) = \min\{|x_j|, \frac{1}{\sqrt 2}
|x_j-x_k|: j\in \SP, \ k\in\SP^{\rm c}\},
\quad \l_{\SP}(\bx) = \min\{1, \dc_{\SP}(\bx)\}.
\end{align}
The paper \cite{FS2021} contains bounds for cluster derivatives $\SD_\SP^m\psi$
(see e.g. \cite[Proposition 1.10]{FS2021}) outside the set $\Sigma_\SP$, depending explicitly 
on the distance $\dc_\SP$. 
We need a generalization of this result to cluster sets 
$\BSP = \{\SP_1, \SP_2, \dots , \SP_M\}$. Define
\begin{align*}
\Sigma_{\BSP} = \cup_{j=1}^M \Sigma_{\SP_j}
\end{align*}
and 
\begin{align}\label{eq:dclusterset}
\dc_{\BSP}(\bx) = \min_j \dc_{\SP_j}(\bx), \quad \l_{\BSP}(\bx) = \min\{1, \dc_{\BSP}(\bx)\}.
\end{align}
The presence of the factor $1/\sqrt2$ in the definition \eqref{eq:dq} is convenient since 
it ensures that the function $\dc_{\BSP}$ is Lipschitz with Lipschitz constant $=1$: 
 
\begin{lem} \label{lem:lip}
For all $\bx, \by\in\R^{3N}$ we have 
\begin{align}\label{eq:lip}
|\dc_{\BSP}(\bx) - \dc_{\BSP}(\by)|\le |\bx-\by|. 
\end{align}
The same inequality holds for the function $\l_{\BSP}(\bx)$.
\end{lem} 

\begin{proof} It suffices to prove the inequality just for one cluster $\SP$. 
Let $j\in \SP$, $k\in \SP^{\rm c}$, and estimate:
\begin{align*}
|x_j|\le &\ |y_j| + |x_j-y_j|\le |y_j|+ |\bx-\by|,\\
|x_j-x_k|\le &\ |y_j-y_k| + |y_j-x_j|+|y_k-x_k|\le |y_j-y_k|+ \sqrt2\, |\bx-\by|.
\end{align*} 
For the last inequality we used the elementary fact that 
if $a^2+b^2\le c^2$ for some positive $a, b$ and $c$, then $a+b\le \sqrt2\,c$. 
Taking the minimum \eqref{eq:dq}, we get 
$\dc_\SP(\bx)\le \dc_{\SP}(\by) + |\bx-\by|$, which  gives \eqref{eq:lip}. 
The function $\l_{\BSP}$ trivially satisfies the same inequality. 
\end{proof}
  
%
Our objective is to find $\plainL\infty$-bounds for the 
cluster derivatives $\SD_{\BSP}^{\bm}\psi$ 
outside the set $\Sigma_{\BSP}$. Let us fix 
$\bx_0\in \R^{3N}\setminus\Sigma_{\BSP}$, and denote 
\begin{align*}
\l := \l_{\BSP}(\bx_0). 
\end{align*}  
Observe that for any $R\in (0, 1)$ the inclusion holds:  
$B(\bx_0, R\l)\subset \R^{3N}\setminus \Sigma_{\BSP}$.  
Indeed, by Lemma \ref{lem:lip},  
\begin{align*}
|\l_{\BSP}(\bx) - \l|\le |\bx-\bx_0|<R\l,\quad \bx\in B(\bx_0, R\l),
\end{align*}
which implies that 
\begin{align}\label{eq:lip2}
0 <(1-R)\l \le \l_{\BSP}(\bx)\le (1+R)\l,\quad \textup{for all}\quad \bx\in B(\bx_0, R\l),
\end{align} 
and hence proves the claim. 
In the next theorem we estimate the derivatives $\SD_{\BSP}^\bm\psi$ in the ball $B(\bx_0, R\l)$.

\begin{thm}\label{thm:clusterpsi}
Let $\bx_0\in \R^{3N}\setminus\Sigma_{\BSP}$
and let $\l = \l_{\BSP}(\bx_0)$. 
Then for any $r, R$ such that $0 < r < R< 1$, and for all 
$\bm\in\mathbb N_0^{3M}$, 
the cluster derivatives $\SD_{\BSP}^\bm \psi$ belong to $\plainC1\big(\overline{B(\bx_0, R\l)}\big)$. 
Moreover, if $|\bm|+k\ge 1$ with $ k = 0, 1$, then  
\begin{align}\label{eq:clusterpsi}
\|\SD_{\BSP}^{\bm}\nabla^k\psi\|_{\plainL{\infty}(B_{}(\bx_0, r\l))}
\lesssim \l^{1-|\bm|- k} 
\big(\|\psi\|_{\plainL{\infty}(B_{}(\bx_0, R\l))} 
+ \|\nabla\psi\|_{\plainL{\infty}(B_{}(\bx_0, R\l))}\big),
\end{align}
with an implicit constant depending on $r, R$, but independent of 
$\bx_0\in\R^{3N}\setminus\Sigma_{\BSP}$. 
\end{thm}

For a single cluster $\SP$  
the bound \eqref{eq:clusterpsi}, even for arbitrary 
$\plainL{p}$-norms with $p\in (1, \infty]$ 
was found in \cite[Proposition 1.10]{FS2021}.
The proof in \cite{FS2021} does not immediately generalize 
to arbitrary cluster sets. 
As in \cite{FS2021} we use the reduction to the equation \eqref{eq:jastrow}. 
In \cite{FS2021} the choice of the Jastrow factor depended on the cluster $\SP$, 
whereas we use the standard function \eqref{eq:F} independent of the cluster set $\BSP$.  
To check that the coefficients in the equation \eqref{eq:jastrow} satisfy 
the conditions of Theorem \ref{thm:reg2}, we start with the following elementary observation. 

\begin{lem}\label{lem:diff}
Let $g \in\plainC\infty(\R^3\setminus \{0\})$ be such that 
\begin{align*}
|\p^m g(x)|\lesssim |x|^{-|m|}, \quad \textup{for all}\quad m\in \mathbb N_0^3. 
\end{align*}  
Then for any cluster set $\BSP$ and any $j, k = 1, 2, \dots, N$, we have 
\begin{align*}
|\SD_{\BSP}^\bm g(x_j)| + |\SD_{\BSP}^\bm g(x_j-x_k)|\lesssim \dc_{\BSP}(\bx)^{-|\bm|},
\end{align*}
for all $\bx\in \R^{3N}\setminus\Sigma_{\BSP}$.
\end{lem}

\begin{proof} 
For simplicity we prove the lemma for one cluster, which is denoted $\SP$. 
Assume that $m\in\mathbb N_0^3$, $|m|\ge 1$. 
It is clear that $D_\SP^m g(x_j-x_k) = 0$ if $j, k\in\SP$ or $j, k\in \SP^{\rm c}$. 
If $j\in \SP, k\in \SP^{\rm c}$ or  $k\in \SP, j\in \SP^{\rm c}$, then 
\begin{align*}
|\SD_\SP^m g(x_j-x_k)|
= |\p_x^m g(x)|_{x = x_j-x_k}\lesssim  |x_j-x_k|^{-|m|}\lesssim \dc_\SP(\bx)^{-|m|},  
\end{align*} 
as claimed. In the same way we get the required estimate for $\SD_{\SP}^m g(x_j)$.
\end{proof}

\begin{cor}
Let $F$ be the function \eqref{eq:F}. Then for any cluster set $\BSP$ we have 
for all $\bx\notin \Sigma_{\BSP}$ and all $\bm$ that 
\begin{align}\label{eq:nabla}
\big|\SD_{\BSP}^\bm \nabla F(\bx)\big| 
+ \big|\SD_{\BSP}^\bm |\nabla F(\bx)|^2\big|\lesssim \dc_{\BSP}(\bx)^{-|\bm|}.
\end{align}
If $|\bm|\ge 1$, then  
\begin{align}\label{eq:ef}
\big|\SD_{\BSP}^\bm e^{F(\bx)}\big| \lesssim  \l_{\BSP}(\bx)^{1-|\bm|}.
\end{align} 
\end{cor}

\begin{proof} Since $F$ 
is a sum of terms of the type $|x_j-x_k|$ and $|x_j|$, the required bound for $\nabla F$ 
immediately follows from the definition \eqref{eq:F} and Lemma 
\ref{lem:diff}. For $|\nabla F|^2$ we use the Leibniz rule, which leads again to \eqref{eq:nabla}.

To prove the bound \eqref{eq:ef} observe that the derivative $\SD_{\BSP}^{\bm}e^F$ represents as a sum of finitely many terms of the form 
\begin{align*}
(\SD_{\BSP}^{\bk_1}F\big)^{n_1}\, (\SD_{\BSP}^{\bk_2}F\big)^{n_2}\cdots 
(\SD_{\BSP}^{\bk_s}F\big)^{n_s} \, e^{F},
\end{align*}
where $1 \le  |\bk_1| < |\bk_2| < \cdots < |\bk_s|\le |\bm|$, 
$n_j\ge 1$, and $|\bk_1|\,n_1 + |\bk_2|\,n_2 + \cdots +|\bk_s|\,n_s = |\bm|$. 
Each such term estimates by 
\begin{align*}
\dc_{\BSP}(\bx)^{n_1(1-|\bk_1|) + n_2(1-|\bk_2|) + \cdots n_s(1-|\bk_1|) } 
= \dc_{\BSP}(\bx)^{n_1+n_2+\dots + n_s - |\bm|}
\le \l_{\BSP}(\bx)^{1-|\bm|},
\end{align*}
as required.
\end{proof}

\begin{proof}[Proof of Theorem \ref{thm:clusterpsi}] 
In view of \eqref{eq:jastrow}, the function $\phi = e^{-F}\psi$ 
satisfies the equation \eqref{eq:mod} where
\begin{align*}
\ba = -2\nabla F,\quad b = V- \Delta F - |\nabla F|^2 - E.
\end{align*}
Let us fix a number $R_1\in (R, 1)$. 
By virtue of \eqref{eq:z} and \eqref{eq:nabla}, thus defined coefficients $\ba$ and $b$ 
satisfy the bound
\begin{align*}
|\SD_{\BSP}^\bm\ba(\bx)| + |\SD_{\BSP}^\bm b(\bx)|\lesssim 1+ \dc_{\BSP}(\bx)^{-|\bm|}
\lesssim \l_{\BSP}(\bx)^{-|\bm|}\lesssim \l^{-|\bm|},\quad \bx\in B(\bx_0, R_1\l),
\end{align*}
where $\l = \l_{\BSP}(\bx_0)$. 
For the last inequality we have used \eqref{eq:lip2}.  
Thus the condition \eqref{eq:cd} is fulfilled 
 with $\ell = \l\le 1$. Consequently, by Theorem \ref{thm:reg2}, 
$\SD_{\BSP}^\bm \phi\in \plainC1(\overline{B(\bx_0, r\l)})$ for all $r <R_1$. 
Moreover, if $|\bm|+k\ge 1$, where $k = 0, 1$, then for  $r < R$ we have 
\begin{align}\label{eq:regphi}
\|\nabla^k \SD_{\BSP}^{\bm} \phi\|_{\plainL\infty(B(\bx_0, r\l))}
\lesssim \l^{1-|\bm|-k}\, \big(\l\,\|\phi\|_{\plainL\infty(B(\bx_0, R\l))}
+ \|\nabla \phi\|_{\plainL\infty(B(\bx_0, R\l))}\big).
\end{align}
Now we need to replace $\phi$ with the function $\psi = e^{F}\phi$. Let us prove \eqref{eq:clusterpsi} 
with $k=0$, $|\bm|\ge 1$. By the Leibniz rule,
\begin{align*}
\SD_{\BSP}^\bm \big(e^{F}\phi\big)
= &\ \sum_{0\le\bq\le \bm} {\bm\choose\bq}\SD_{\BSP}^{\bm-\bq} \big(e^{F}\big)  \SD_{\BSP}^{\bq}\phi\\
= &\ \sum_{\substack{0\le\bq\le \bm \\ 0<|\bq| < |\bm|}}
{\bm\choose\bq}\SD_{\BSP}^{\bm-\bq} \big(e^{F}\big)  \SD_{\BSP}^{\bq}\phi
+ \SD_{\BSP}^{\bm} \big(e^{F}\big) \phi
+ e^{F} \SD_{\BSP}^{\bm}\phi.
\end{align*}
By \eqref{eq:ef} and \eqref{eq:regphi}, for $\bx\in B(\bx_0, r\l)$ the first sum is bounded by 
\begin{align*}
\l^{2-|\bm|}\, \big(\l\|\phi\|_{\plainL\infty(B(\bx_0, R\l))}
+ \|\nabla \phi\|_{\plainL\infty(B(\bx_0, R\l))}\big).
\end{align*}
The second term is estimated with the help of \eqref{eq:ef} by
\begin{align*}
\l^{1-|\bm|}\, \|\phi\|_{\plainL\infty(B(\bx_0, R\l))}.
\end{align*}
Using \eqref{eq:regphi} the third term is estimated by 
\begin{align*}
\l^{1-|\bm|}\, \big(\l\,\|\phi\|_{\plainL\infty(B(\bx_0, R\l))}
+ \|\nabla \phi\|_{\plainL\infty(B(\bx_0, R\l))}\big).
\end{align*}
Consequently, 
\begin{align}\label{eq:regpsi}
\|\SD_{\BSP}^{\bm} \psi\|_{\plainL\infty(B(\bx_0, r\l))}
\lesssim \l^{1-|\bm|}\, \big(\|\phi\|_{\plainL\infty(B(\bx_0, R\l))}
+ \|\nabla \phi\|_{\plainL\infty(B(\bx_0, R\l))}\big).
\end{align}
It remains to note that in view of \eqref{eq:Finf}, 
\begin{align*}
|\phi|\lesssim |\psi|,\quad |\nabla\phi|\lesssim |\psi| + |\nabla\psi|, 
\end{align*}
so that the right-hand side of \eqref{eq:regpsi} is estimated by the right-hand side 
of \eqref{eq:clusterpsi} with $k = 0$.

The case $k = 1$ is done in the same way.
\end{proof}

Later on we use Theorem \ref{thm:clusterpsi} in a slightly different form:

\begin{cor}\label{cor:clusterpsi}
For all $\bm\in \mathbb N_0^{3M}$ the cluster derivatives 
$\SD_{\BSP}^\bm \psi$ belong to $\plainC1\big(\R^{3N}\setminus\Sigma_{\BSP}\big)$, 
and under the condition $|\bm|+k\ge 1$ with $k= 0, 1$, 
for all $\bx\in \R^{3N}\setminus\Sigma_{\BSP}$ and all $R>0$, we have 
\begin{align}
|\SD_{\BSP}^{\bm}\nabla^k\psi(\bx)|
\lesssim &\ \l_{\BSP}(\bx)^{1-|\bm|- k} f_{\infty}(\bx; R),\label{eq:clustmod}\\
f_\infty(\bx; R) = &\ 
\|\psi\|_{\plainL{\infty}(B(\bx, R))} 
+ \|\nabla\psi\|_{\plainL{\infty}(B(\bx, R))}.\label{eq:finfinity}
\end{align} 
The implicit constant in \eqref{eq:clustmod} is independent of $\psi$ and $\bx$, but 
may depend on $R$. 
\end{cor} 
 
\begin{proof}
Let $R_1 = \min\{1/2, R\}$. Then it follows from \eqref{eq:clusterpsi} that 
\begin{align*}
|\SD_{\BSP}^{\bm}\nabla^k\psi(\bx)|
\lesssim &\ \l(\bx)^{1-|\bm|- k} f_{\infty}(\bx; R_1\l(\bx))
\le \l(\bx)^{1-|\bm|- k} f_{\infty}(\bx; R).   
\end{align*}
Here we have used the fact that $\l(\bx)\le 1$. This completes the proof.
\end{proof} 

\section{Auxiliary integral bounds} \label{sect:aux}

Here we derive several integral bounds that are instrumental in the proof of the main result in Section 
\ref{sect:estim}. 

Let $\psi\in\plainL2(\R^{3N})$ be an eigenfunction.  Along with 
the notation \eqref{eq:finfinity} it is 
convenient to introduce for arbitrary $R>0$ also 
\begin{align*}
f_2(\bx)= f_2(\bx; R) 
= \|\psi\|_{\plainL{2}(B(\bx, R))}. 
\end{align*} 
According to \eqref{eq:grpsi}, 
for any $R>0$ we have
\begin{align}\label{eq:inf_2}
f_\infty(\bx; R) \lesssim f_2(\bx; 2R), 
\end{align}
with an implicit constant depending on $R$.
 
The next lemma provides bounds for integrals involving the functions $f_\infty(\bx)$ and $f_2(\bx)$.  
Recall that the function $\CM_\varepsilon(x, y, \hat\bx)$  is defined in \eqref{eq:phibm}, 
and the density $\rho(x)$ -- in \eqref{eq:dens}
 
\begin{lem}
Let $x, y \in\R^3$. Then 
\begin{align}\label{eq:infinf} 
\int f_\infty(x, \hat\bx; R) f_\infty(y, \hat\bx; R) d\hat\bx
\lesssim 
\big(\|\rho\|_{\plainL1(B(x, 2R))}\big)^{\frac{1}{2}}
\big(\|\rho\|_{\plainL1(B(y, 2R))}\big)^{\frac{1}{2}}.
%
%
%
\end{align}
For any $G\in \plainL1(\R^3)$ we have  
\begin{align}\label{eq:couinf1}
\int \big[|G(x_j-x_k)| +|G(t-x_k)| + &\ |G(x_k)|\big]
 f_\infty(x, \hat\bx; R)  f_\infty(y, \hat\bx; R)\, d\hat\bx
\notag\\
&\ \lesssim  
\|G\|_{\plainL1(\R^3)}
\big(\|\rho\|_{\plainL1(B(x, 2R))}\big)^{\frac{1}{2}}
\big(\|\rho\|_{\plainL1(B(y, 2R))}\big)^{\frac{1}{2}},
%
%
\end{align} 
for all $j, k = 2, 3, \dots, N, j\not = k$, and $t\in\R^3$. 
In particular,
\begin{align}\label{eq:couinf2}
\int \CM_\varepsilon(x, y, \hat\bx) 
 f_\infty(x, \hat\bx; R) f_\infty(y, \hat\bx; R)\, d\hat\bx
 \lesssim  \varepsilon^3 
\big(\|\rho\|_{\plainL1(B(x, 2R))}\big)^{\frac{1}{2}}
\big(\|\rho\|_{\plainL1(B(y, 2R))}\big)^{\frac{1}{2}}. 
\end{align} 
The implicit constants in the 
bounds \eqref{eq:infinf}, \eqref{eq:couinf1} and  \eqref{eq:couinf2} 
depend on $R$, but are independent of $x, y, t\in\R^3$. 
\end{lem}

\begin{proof} 
Due to the Schwarz inequality, it suffices to estimate the integrals for $x = y$. 
For $u\in \R^d$ and $R>0$ denote by $\1^{(d)}_{u, R}(x), x\in\R^d,$ the indicator function 
of the ball $B(u, R)\subset \R^d$.  

Using \eqref{eq:inf_2} we get
\begin{align*}
\int \big(f_\infty(x, \hat\bx; R)\big)^2 d\hat\bx
\lesssim &\ \int \big(f_2(x, \hat\bx; 2R)\big)^2 d\hat\bx\\[0.2cm]
= &\ \int \int |\psi(\bz)|^2 \1^{(3N)}_{(x, \hat\bx), 2R}(\bz) \,d\bz \,d\hat\bx\\
= &\ \int  |\psi(\bz)|^2 \int\1^{(3N)}_{\bz, 2R}(x, \hat\bx) \,d\hat\bx \,d\bz. 
\end{align*}
Observe that  
\begin{align*}
\1^{(3N)}_{\bz, 2R}(x, \hat\bx) 
\le  \1^{(3)}_{z_1, 2R}(x)
\1^{(3N-3)}_{\hat\bz, 2R}(\hat\bx) 
= \1^{(3)}_{x, 2R}(z_1)
\1^{(3N-3)}_{\hat\bz, 2R}(\hat\bx),
\end{align*}
so the integral does not exceed 
\begin{align*}
\int  |\psi(\bz)|^2 \1^{(3)}_{x, 2R}(z_1) 
\bigg(\int
\1^{(3N-3)}_{\hat\bz, 2R}(\hat\bx)\,d\hat\bx\bigg) \,d\bz
  \lesssim R^{3N-3} 
\int\limits_{|z-x|<2R} \rho(z) dz.  
%
%
\end{align*}
This proves \eqref{eq:infinf}. 

Proof of \eqref{eq:couinf1}.  
Again it suffices to estimate the integral for $x = y$:
\begin{align*}
\int 
|G(x_j-x_k)|&\ 
\big(f_\infty(x, \hat\bx; R)\big)^2\, d\hat\bx
\lesssim \int   
|G(x_j-x_k)|  \big(f_2(x, \hat\bx; 2R)\big)^2\, d\hat\bx\\[0.2cm]
= &\ \int  |\psi(\bz)|^2 \int   
|G(x_j-x_k)| \1^{(3N)}_{\bz, 2R}(x, \hat\bx) \,d\hat\bx \,d\bz. 
\end{align*}
Represent $\hat\bx = (x_j, \tilde\bx_{j})$ 
with $\tilde\bx_j\in \R^{3N-6}$, as defined in \eqref{eq:xtilde},   
and estimate:
\begin{align*}
\1^{(3N)}_{\bz, 2R}(x, \hat\bx)
\le  \1^{(3)}_{z_1, 2R}(x)\, \1^{(3N-6)}_{\tilde\bz_j, 2R}(\tilde\bx_j)  
= \1^{(3)}_{x, 2R}(z_1)\, \1^{(3N-6)}_{\tilde\bz_j, 2R}(\tilde\bx_j).
\end{align*}
Consequently, the integral estimates by 
\begin{align*}
\int  |\psi(\bz)|^2 \1^{(3)}_{x, 2R}(z_1)&\ \int \1_{B(\tilde\bz_j, 2R )}(\tilde\bx_j)
\bigg[\int |G(x_j-x_k)|\, dx_j  \bigg]\,d\tilde\bx_j \,d\bz\\
 = &\ 
\|G\|_{\plainL1(\R^3)}
\int\limits_{|z_1-x|<2R}  |\psi(\bz)|^2 \int
\1^{(3N-6)}_{\tilde\bz_j, 2R}(\tilde\bx_j)\,d\tilde\bx_j \,d\bz\\
\lesssim &\ 
\|G\|_{\plainL1(\R^3)} R^{3N-6}
\int\limits_{|z-x|<2R}\rho(z) dz, 
%
%
\end{align*}
as claimed. The bounds with $G(t-x_k)$ and $G(x_k)$ are proved in the same way.

The bound \eqref{eq:couinf2} follows from \eqref{eq:couinf1} 
with $G(s) = \xi(N\varepsilon^{-1}|s|)$, see \eqref{eq:xi} for the definition of the function $\xi$.
\end{proof}

Let us now apply the obtained bounds to integrals involving the distance function 
$\l_{\SQ}(\bx)$ with an arbitrary cluster $\SQ$. 
Recall that $\l_{\SQ}(\bx)$, $\widehat X_\SQ(t;\varepsilon)$ and $\widehat T_{\SQ}(\d)$ are 
defined in \eqref{eq:dq}, \eqref{eq:hatxp} and \eqref{eq:hattp} respectively, and the function 
$h_a$ is defined in \eqref{eq:h}.  

\begin{lem} \label{lem:Xint}
Let $a\ge 0$ and $\varepsilon >0$. 
Let $\SQ$ be an arbitrary cluster.  
Then for any $R>0$, 
\begin{align}\label{eq:Xint}
\int\limits_{\widehat X_\SQ(t; \varepsilon)\cap \widehat T_{\SQ}(\varepsilon)}  
\l_{\SQ}(t, \hat\bx)^{-a} 
f_\infty(x, \hat\bx; R) &\ f_\infty(y, \hat\bx; R) d\hat\bx\notag\\
\lesssim &\ \big(1+|t|^{-a} + h_{a+2}(\varepsilon)\big)\,
\big(\|\rho\|_{\plainL1(B(x, 2R))}\big)^{\frac{1}{2}}
\big(\|\rho\|_{\plainL1(B(y, 2R))}\big)^{\frac{1}{2}},
%
%
\end{align}
uniformly in $x, y, t\in\R^3$. If $1\notin \SQ$, then the term $|t|^{-a}$ is absent. 
The implicit constant in \eqref{eq:Xint} may depend on $R$, but is independent of $\varepsilon$.
\end{lem}
 
\begin{proof}  
By the definition of $\l_\SQ$ we have 
\begin{align}\label{eq:notin}
\l_\SQ(t, \hat\bx)^{-a}
\le  &\ 1 + \sum_{k\in\SQ}|x_k|^{-a}\notag\\
&\ \sum_{k\in\SQ}\, |t-x_k|^{-a} + \sum_{j\in\SQ, k\in (\SQ^{\rm c})^*} 
|x_j-x_k|^{-a},\quad \textup{if}\quad 1\in\SQ^{\rm c},  
\end{align}
and 
\begin{align}\label{eq:in}
\l_\SQ(t, \hat\bx)^{-a}
\le  1 + &\ |t|^{-a} + \sum_{k\in\SQ^*}|x_k|^{-a} \notag\\ &\ 
+  \sum_{k\in\SQ^{\rm c}} 
|t-x_k|^{-a}
 + \sum_{j\in\SQ^*, k\in \SQ^{\rm c}} |x_j-x_k|^{-a},\quad \textup{if}\quad 1\in\SQ.  
\end{align}
Let us estimate the contributions from each of the summands. 
Assume first that $1\notin\SQ$ (i.e. $\l_\SQ$ satisfies \eqref{eq:notin}) 
and estimate the integral
\begin{align}\label{eq:jk}
F_{jk}(t, x, y) = \int\limits_{\widehat X_\SQ(t; \varepsilon)\cap \widehat T_{\SQ}(\varepsilon)}  
|x_j-x_k|^{-a} 
f_\infty(x, \hat\bx; R) &\ f_\infty(y, \hat\bx; R)\, d\hat\bx,
%
%
\end{align}
%
%
%
%
%
for an arbitrary fixed pair $j\in\SQ, k\in(\SQ^{\rm c})^*$. 
Since $j, k\ge 2$, in view of the definition 
\eqref{eq:hatxp} of $\widehat X_\SQ(t, \varepsilon)$, 
we have 
\begin{align*}
\widehat X_\SQ(t; \varepsilon)\subset \{\hat\bx\in\R^{3N-3}: |x_j-x_k| > \varepsilon\}.
\end{align*}
%
%
Therefore, 
%
%
\begin{align*}
F_{jk}(t, x, y)\le \int\limits_{|x_j-x_k|> \varepsilon}  
 |x_j-x_k|^{-a}   \,
f_\infty(x, \hat\bx; R)\, f_\infty(y, \hat\bx; R)\, d\hat\bx.
\end{align*}
If $\varepsilon\ge 1$, then by \eqref{eq:infinf} we have 
\begin{align}\label{eq:ege1}
F_{jk}(t, x, y)\lesssim \big(\|\rho\|_{\plainL1(B(x, 2R))}\big)^{\frac{1}{2}}
\big(\|\rho\|_{\plainL1(B(y, 2R))}\big)^{\frac{1}{2}}.
\end{align}
If $\varepsilon<1$, then  
\begin{align*}
F_{jk}(t, x, y)\le &\ \int\limits_{|x_j-x_k|> 1}   
f_\infty(x, \hat\bx; R)\, f_\infty(y, \hat\bx; R)\, d\hat\bx\\
&\qquad + 
\int G(x_j-x_k) \,
f_\infty(x, \hat\bx; R)\, f_\infty(y, \hat\bx; R)\, d\hat\bx,
\end{align*}
where 
\begin{align}\label{eq:gs}
G(s) = \1_{\{\varepsilon<|s|<1\}}(s)|s|^{-a},\quad s\in\R^3.
\end{align}
For the first integral use \eqref{eq:infinf} again. Since 
$\|G\|_{\plainL1(\R^3)}\lesssim h_{a+2}(\varepsilon)$ (see \eqref{eq:h}) the second integral 
is bounded by 
\begin{align*}
h_{a+2}(\varepsilon)\big(\|\rho\|_{\plainL1(B(x, 2R))}\big)^{\frac{1}{2}}
\big(\|\rho\|_{\plainL1(B(y, 2R))}\big)^{\frac{1}{2}},
\end{align*}
in view of of \eqref{eq:couinf1}. 
Thus the integral \eqref{eq:jk}, and hence the contribution from 
the last term in \eqref{eq:notin} as well,  satisfies the bound \eqref{eq:Xint}.  

%
%
In the same way, using \eqref{eq:couinf1} we derive the bound \eqref{eq:Xint} 
%
%
for the integrals containing the remaining terms in \eqref{eq:notin}. 
Let us estimate, for example, the integral 
\begin{align*}
F_k(t, x, y) = \int\limits_{\widehat X_\SQ(t; \varepsilon)\cap \widehat T_{\SQ}(\varepsilon)}  
|x_k|^{-a} \,
f_\infty(x, \hat\bx; R)\, f_\infty(y, \hat\bx; R)\, d\hat\bx,
\end{align*}
for an arbitrary fixed $k\in\SQ$. By the definition \eqref{eq:hattp}, 
\begin{align*}
\widehat T_\SQ(\varepsilon)\subset \{\hat\bx\in\R^{3N-3}: |x_k| > \varepsilon\}, 
\end{align*}
so that 
\begin{align*}
F_k(t, x, y) \le \int\limits_{|x_k| > \varepsilon}
|x_k|^{-a} \,
f_\infty(x, \hat\bx; R)\, f_\infty(y, \hat\bx; R)\, d\hat\bx,
\end{align*} 
As in the case of the integral \eqref{eq:jk}, if $\varepsilon \ge 1$, then $F_k$ satisfies the bound \eqref{eq:ege1}. 
If $\varepsilon<1$, then 
\begin{align*}
F_{k}(t, x, y)\le &\ \int\limits_{|x_k|> 1}   
f_\infty(x, \hat\bx; R)\, f_\infty(y, \hat\bx; R)\, d\hat\bx\\
&\qquad + 
\int G(x_k) \,
f_\infty(x, \hat\bx; R)\, f_\infty(y, \hat\bx; R)\, d\hat\bx,
\end{align*}
where $G$ is given by \eqref{eq:gs}. Arguing as for 
$F_{jk}$ above, we conclude that $F_k$ satisfies \eqref{eq:Xint} as well. 
%
%
%
%
%
%
%
%
%
%
This proves \eqref{eq:Xint} for the case $1\notin \SQ$.
%
%

Suppose now that $1\in\SQ$ so that $\l_{\SQ}$ satisfies \eqref{eq:in}.    
The sums on the right-hand side of \eqref{eq:in} are similar to those in \eqref{eq:notin} and 
are treated as in the first part of the proof, and hence they lead to the estimate \eqref{eq:Xint}. 
The only term which is new compared to \eqref{eq:notin} is $|t|^{-a}$.
Using \eqref{eq:infinf} we estimate its contribution by 
\begin{align*}
|t|^{-a}\int    f_\infty(x, \hat\bx; R)\, f_\infty(y, \hat\bx; R)\, d\hat\bx
\lesssim \, |t|^{-a}\, 
\big(\|\rho\|_{\plainL1(B(x, 2R))}\big)^{\frac{1}{2}}
\big(\|\rho\|_{\plainL1(B(y, 2R))}\big)^{\frac{1}{2}}.
\end{align*}
This completes the proof of \eqref{eq:Xint}.
\end{proof} 
 
Lemma \ref{lem:Xint} has a useful corollary that 
will be crucial in the proof of Theorem \ref{thm:derbounds}. 
Let $\Phi = \Phi(x, y, \hat\bx; \varepsilon)$ be an extended 
cut-off as defined in \eqref{eq:Phi}, 
and let $\SP = \SP(\varepsilon)$ and $\SfS = \SfS(\varepsilon)$ be the clusters for the 
admissible cut-offs associated with $\Phi$. 
For all $a\ge 0$ define 
\begin{align}\label{eq:ja}
\begin{cases}
{\mathcal J}_a^{(1)}(x, y; \varepsilon, R)
= &\ 
\int\limits_{\supp_0\Phi(x, y, \ \cdot\ ; \varepsilon)} \big(\l_{\{\SfS^*,\SP\}}
(x, \hat\bx)\big)^{-a} f_\infty(x, \hat\bx; R)f_\infty(y, \hat\bx; R)\, d\hat\bx, \\[0.5cm]
{\mathcal J}_a^{(2)}(x, y; \varepsilon, R) 
= &\ \int\limits_{\supp_0\Phi(x, y, \ \cdot\ ; \varepsilon)} \big(\l_{\{\SfS,\SP^*\}}
(y, \hat\bx)\big)^{-a} f_\infty(x, \hat\bx; R)f_\infty(y, \hat\bx; R)\, d\hat\bx, \\[0.5cm]
{\mathcal J}_a^{(3)}(x, y; \varepsilon, R) 
= &\ \varepsilon^{-a}\int\,
\CM_\varepsilon(x, y; \hat\bx) 
f_\infty(x, \hat\bx; R)f_\infty(y, \hat\bx; R)\, d\hat\bx. 
\end{cases}
\end{align}   
 
\begin{lem}\label{lem:ja}
Suppose that for some $\varepsilon>0$ 
we have $|x|>\varepsilon$, $|y|>\varepsilon$ and $|x-y| > \varepsilon$. Then 
for all $R>0$, 
\begin{align}\label{eq:ja3}
{\mathcal J}_a^{(1)}(x, y; \varepsilon, R)
+ &\ {\mathcal J}_a^{(2)}(x, y; \varepsilon, R)
+ {\mathcal J}_a^{(3)}(x, y; \varepsilon, R)\notag\\
\lesssim &\ \big(1+|x|^{-a} + |y|^{-a} + h_{a+2}(\varepsilon)\big)\,
\big(\|\rho\|_{\plainL1(B(x, 2R))}\big)^{\frac{1}{2}}
\big(\|\rho\|_{\plainL1(B(y, 2R))}\big)^{\frac{1}{2}}.
%
%
\end{align} 
The implicit constant in \eqref{eq:Xint} may depend on $R$, but is independent of $\varepsilon$.
\end{lem} 
 
\begin{proof} 
By Proposition \ref{prop:empty} we may assume that 
$\SP^*\subset \SfS^{\rm c}$.   

First estimate $\mathcal J_a^{(1)}(x, y; \varepsilon, R)$.  
By the definition \eqref{eq:dclusterset}, we have 
\begin{align*}
\l_{\{\SfS^*, \SP\}}(x, \hat\bx)^{-1}\le \l_{\SfS^*}(x, \hat\bx)^{-1} + \l_{\SP}(x, \hat\bx)^{-1}.
\end{align*}
Furthermore, since $\SP^*\subset \SfS^{\rm c}$, we can use 
\eqref{eq:phimux} and \eqref{eq:phihat}: 
\begin{align*}
\supp_0\Phi(x, y; \ \cdot\ ; \varepsilon)\subset \widehat X_\SP(x; \varepsilon(4N)^{-1})
\cap \widehat T_{\SP^*}(\varepsilon/2)\quad \textup{and}\\
\supp_0\Phi(x, y, \ \cdot\ ; \varepsilon)\subset \widehat X_{\SfS^*}(x; \varepsilon(4N)^{-1})
\cap \widehat T_{\SfS^*}(\varepsilon/2).
\end{align*}
These lead to the bound 
\begin{align*}
{\mathcal J}_a^{(1)}(x, y; \varepsilon, R)
\lesssim &\ \int\limits_{\widehat X_\SP(x; \varepsilon(4N)^{-1})
\cap \widehat T_{\SP}(\varepsilon/2)}    \l_{\SP}(x, \hat\bx)^{-a} 
f_\infty(x, \hat\bx; R) f_\infty(y, \hat\bx; R) d\hat\bx\\
 &\ + \int\limits_{\widehat X_{\SfS^*}(x; \varepsilon(4N)^{-1})\cap 
 \widehat T_{\SfS^*}(\varepsilon/2)}  \l_{\SfS^*}(x, \hat\bx)^{-a} 
f_\infty(x, \hat\bx; R)f_\infty(y, \hat\bx; R) d\hat\bx.  
\end{align*}
By \eqref{eq:Xint} with $t = x$, each of these integrals 
is bounded by 
\begin{align*}
\big(1+|x|^{-a} + h_{a+2}
(\varepsilon)\big)\,
\big(\|\rho\|_{\plainL1(B(x, 2R))}\big)^{\frac{1}{2}}
\big(\|\rho\|_{\plainL1(B(y, 2R))}\big)^{\frac{1}{2}},
%
%
\end{align*}
which implies \eqref{eq:ja3}.  
The integral ${\mathcal J}_a^{(2)}(x, y; \varepsilon, R)$ is estimated in the same way. 

The integral $\mathcal J_a^{(3)}(x, y; \varepsilon, R)$ is estimated with the help of 
\eqref{eq:couinf2}: 
\begin{align*}
\mathcal J_a^{(3)}(x, y; \varepsilon, R)\lesssim &\ \varepsilon^{3-a}
\,
\big(\|\rho\|_{\plainL1(B(x, 2R))}\big)^{\frac{1}{2}}
\big(\|\rho\|_{\plainL1(B(y, 2R))}\big)^{\frac{1}{2}}\\
%
%
\lesssim &\ h_{a+2}(\varepsilon) \,
\big(\|\rho\|_{\plainL1(B(x, 2R))}\big)^{\frac{1}{2}}
\big(\|\rho\|_{\plainL1(B(y, 2R))}\big)^{\frac{1}{2}}.
%
%
\end{align*}
Combining the obtained bounds we complete the proof of \eqref{eq:ja3}.
\end{proof}

Now we are in a position to prove Theorem \ref{thm:derbounds}.

\section{Estimates for the density matrix: proof of Theorem \ref{thm:derbounds}
}\label{sect:estim}

To cover both inequalities \eqref{eq:der2} and \eqref{eq:der1} in one proof we 
study the following function:
\begin{align}\label{eq:rp}
\g_{r, p}(x, y) := \p_x^r\p_y^p\g(x, y) 
= \int \p_x^r\psi(x, \hat\bx) \overline{\p_y^p\psi(y, \hat\bx)}\, d\hat\bx, 
\end{align}
where $r, p\in\mathbb N_0^3$ are such that $|r|\le 1$, $|p|\le 1$. 
By Corollary \ref{cor:clusterpsi} the integrand in \eqref{eq:rp} does not exceed  
$f_\infty(x, \hat\bx; R) f_\infty(y, \hat\bx; R)$, see \eqref{eq:finfinity} for the definition of 
$f_\infty(\bx; R)$. Therefore, the estimate \eqref{eq:infinf} entails the bound 
\begin{align}\label{eq:rpest}
|\g_{r, p}(x, y)|\lesssim 
\big(\|\rho\|_{\plainL1(B(x, 2R))}\big)^{\frac{1}{2}}
\big(\|\rho\|_{\plainL1(B(y, 2R))}\big)^{\frac{1}{2}}.
\end{align}
In order to estimate the derivatives of $\g_{r, p}(x, y)$ 
we begin by studying ``local" quantities. For a fixed $\varepsilon>0$ 
let $\Phi = \Phi(x, y, \hat\bx; \varepsilon)$ 
be an extended cut-off as defined in \eqref{eq:Phi}.
Define 
\begin{align}\label{eq:rpphi}
\g_{r, p}(x, y; \Phi, \varepsilon)  
= \int \Phi(x, y, \hat\bx; \varepsilon) \p_x^r\psi(x, \hat\bx) 
\overline{\p_y^p\psi(y, \hat\bx)}\, d\hat\bx. 
\end{align} 

\begin{rem}\label{rem:grp}
The motivation to study the function \eqref{eq:rp} instead of the original one-particle density matrix 
\eqref{eq:gamma} can be explained by the following intuitive argument. 
As we saw earlier (see Corollary \ref{cor:clusterpsi}) higher order cluster 
derivatives of $\psi$ become ever more singular 
near the coalescence set $\Sigma_{\BSP}$ (i.e. for small $\l_{\BSP}$)  
as the order of derivative grows. 
This however does not apply to 
first order derivatives -- together with the function $\psi$ itself, 
they remain bounded near $\Sigma_{\BSP}$. 
Thus, taking derivatives of first order under the integral \eqref{eq:gamma} does not make 
the integral ``worse". At the same time, 
by virtue of the formula 
\begin{align}\label{eq:obv}
\p_x^m\p_y^l \g(x, y) = \p_x^{m-r}\p_y^{l-p}\gamma_{r, p}(x, y),
\end{align}
this reduces the order of the remaining 
derivatives. In particular, when proving \eqref{eq:der2} further in this section, 
where $\g$ is differentiated with respect to 
both variables $x$ and $y$, we use \eqref{eq:obv} with $|p|=|r|=1$, thereby reducing the order of the  derivatives  with respect to $x$ and $y$ by one each.  
On the other hand, in the proof of \eqref{eq:der1} where the differentiation is conducted with respect to one 
of the variables $x$ or $y$ only, we cannot take full  
advantage  of this  
order reduction. Rather, we derive \eqref{eq:der1} 
from \eqref{eq:der2} using the fundamental theorem of calculus which introduces the ``missing" derivative  
artificially.  

\end{rem}

Let us estimate the derivatives of $\g_{r, p}$ for arbitrary $|r|\le 1$, $|p|\le 1$. 

\begin{lem}\label{lem:11}
Let $|r|\le 1$, $|p|\le 1$. 
Suppose that for some $\varepsilon$ we have 
$|x|>\varepsilon$, $|y|>\varepsilon$ and $|x-y| > \varepsilon$. 
Then for all $\a, \b\in \mathbb N_0^3$, and all $R>0$, 
\begin{align}\label{eq:11}
|\p_x^\a \p_y^\b  \g_{r, p}(x,& y, \Phi; \varepsilon)|\notag\\
\lesssim &\ \big( 1 + |x|^{-|\a|-|\b|} + |y|^{-|\a|-|\b|} 
+ h_{|\a|+|\b|+2}(\varepsilon)\big)\,
\big(\|\rho\|_{\plainL1(B(x, R))}\big)^{\frac{1}{2}}
\big(\|\rho\|_{\plainL1(B(y, R))}\big)^{\frac{1}{2}}.
%
%
\end{align} 
with an implicit constant independent of $\psi$, $\varepsilon$, but dependent on $R$.  
\end{lem}

\begin{proof} Throughout the proof for the brevity of notation 
we often omit the dependence on $\varepsilon$.  
Let $\SP$ and $\SfS$ be the clusters for the 
admissible cut-offs $\phi$ and $\mu$ associated with $\Phi$.  By Proposition 
\ref{prop:empty} we may assume that 
$\SfS^*\subset \SP^{\rm c}$, which is equivalent to $\SP^*\subset \SfS^{\rm c}$. 

If $\a = \b = 0$, then \eqref{eq:11} holds because of \eqref{eq:rpest}. Thus we may assume that 
$|\a|+|\b|\ge 1$. Let us 
make the following change of variables 
under the integral $\g_{r, p}(x, y; \Phi)$. Define 
$\hat\bz = (z_2, z_3, \dots, z_N)\in\R^{3N-3}$ by 
\begin{align*}
z_j = 
\begin{cases}
x,\ j\in \SP^*,\\
y,\ j\in \SfS^*,\\
0,\ j\in\SP^{\rm c}\cap\SfS^{\rm c}. 
\end{cases}
\end{align*}
Change the variable in \eqref{eq:rpphi}: $\hat\bx = \hat\bw + \hat\bz$, so that 
\eqref{eq:rpphi} rewrites as
\begin{align*}
\g_{r, p}(x, y; \Phi)  
= &\ \int \Phi(x, y, \hat\bw + \hat\bz) \p_x^r\psi(x, \hat\bw+\hat\bz) 
\overline{\p_y^p\psi(y, \hat\bw + \hat\bz)}\, d\hat\bw.
\end{align*}
For any function $g = g(x, \hat\bx)$ and 
all $l, m_1, m_2\in\mathbb N_0^3$, we have  
\begin{align*}
\p_x^{l}\big(g(x, \hat\bw+\hat\bz)\big)  = 
(\SD_{\SP}^l g)(x, \hat\bw+\hat\bz),\quad  
\p_x^{l}\big(g(y, \hat\bw+\hat\bz)\big)  = 
(\SD_{\SP^*}^l g)(y, \hat\bw+\hat\bz)
\end{align*}
and 
\begin{align*}
\p_x^{m_1}\p_y^{m_2}\big(\Phi(x, y, \hat\bw + \hat\bz)\big)
 = \big(\SD_{x,\SP}^{m_1}\SD_{y, \SfS}^{m_2} \Phi\big)(x, y,  \hat\bw + \hat\bz), 
\end{align*}
where we have used the notation \eqref{eq:clustermul} for the cluster derivatives 
of the function $\Phi$. 
Denote 
\begin{align*}
Z(x, y) = \supp_0\, \Phi(x, y, \,\cdot\,).
\end{align*}
Thus we conclude that $\p_x^\a \p_y^\b\g_{r, p}(x, y; \Phi), \a, \b\in\mathbb N_0^3,$ 
is a linear combination of terms of the form
\begin{align}\label{eq:jmnk}
{\mathcal J}_{\bm, \bn, \bk}(x, y; \varepsilon) = \int_{Z(x, y)}
\big(\SD_{x,\SP}^{m_1}\SD_{y, \SfS}^{m_2} &\ \Phi(x, y,  \hat\bx) \big)\notag\\
&\ \times\big(\SD_{\{\SP, \SfS^*\}}^{\bn} \p_x^r\psi(x, \hat\bx)\big) 
\overline{\big(\SD_{\{\SP^*, \SfS\}}^{\bk}\p_y^p\psi(y, \hat\bx)\big)}\, d\hat\bx
\end{align} 
with 
\begin{align*}
\bm = (m_1, m_2),\ \bn = &\ (n_1, n_2), \bk = (k_1, k_2),\quad \textup{where}\\
|m_1|+|n_1|+|k_1| = |\a|,\quad &\ |m_2|+|n_2|+|k_2| = |\b|. 
\end{align*}
The cluster derivatives of $\Phi$ are defined in 
\eqref{eq:clustermul}. For these derivatives we use the bound \eqref{eq:phib}:
\begin{align}\label{eq:dphi}
|\SD_{x,\SP}^{m_1} \SD_{y,\SfS}^{m_2}\Phi(x, y, \hat\bx)| 
\lesssim 
\begin{cases}
1, \quad \textup{if}\quad |\bm|=0,\\[0.2cm]
\varepsilon^{-|\bm|} \CM_\varepsilon(x, y, \hat\bx),\quad \textup{if}\quad |\bm|\ge 1,
\end{cases}
\end{align}  
For the derivatives of $\psi$ we use Corollary \ref{cor:clusterpsi}: 
\begin{align}\label{eq:dnkpsi}
\begin{cases}
|\SD_{\{\SfS^*,\SP\}}^{\bn}\p_x^r\psi(x, \hat\bx) |
\lesssim &\ \big(\l_{\{\SfS^*,\SP\}}(x, \hat\bx)\big)^{-|\bn|}
f_\infty(x, \hat\bx; R/2),\\[0.2cm]
|\SD_{\{\SfS,\SP^*\}}^{\bk}\p_y^p\psi(y, \hat\bx) |
\lesssim &\ \big(\l_{\{\SfS,\SP^*\}}(y, \hat\bx)\big)^{-|\bk|} f_\infty(y, \hat\bx; R/2).
\end{cases}
\end{align}
In order to avoid cumbersome expressions, in the following calculations we use the notation 
\begin{align*}
\mu(x, \hat\bx) = \l_{\{\SfS^*,\SP\}}(x, \hat\bx),\,   
\tilde\mu(y, \hat\bx) = \l_{\{\SfS,\SP^*\}}(y, \hat\bx).
\end{align*}
Assume first that $\bm = \bold0$. In this case, by virtue of 
\eqref{eq:dphi} and \eqref{eq:dnkpsi}, 
 the integral \eqref{eq:jmnk} satisfies the estimate
\begin{align*}
|{\mathcal J}_{\bold0, \bn, \bk}(x, y; \varepsilon)| \lesssim 
\int_{Z(x, y)} \big(\mu(x,\hat\bx)\big)^{-|\bn|}\,\big(\tilde\mu(y,\hat\bx)\big)^{-|\bk|}
f_\infty(x, \hat\bx; R/2)\,  f_\infty(y, \hat\bx; R/2)\, d\hat\bx.
\end{align*}
By Young's inequality, for all $\bn, \bk$, such that $|\bn|+|\bk| = |\a|+|\b|\ge 1$, we have
\begin{align*}
\big(\mu(x, \hat\bx)\big)^{-|\bn|}
&\ \big(\tilde\mu(y, \hat\bx)\big)^{-|\bk|} \\
&  \le \frac{|\bn|}{|\a|+|\b|}
\big(\mu(x, \hat\bx)\big)^{-|\a| - |\b|}
+ \frac{|\bk|}{|\a|+|\b|}
\big(\tilde\mu(y, \hat\bx)\big)^{-|\a| - |\b|}.
\end{align*}
Consequently, using the notation \eqref{eq:ja} we can estimate,
\begin{align*}
|{\mathcal J}_{\bold0, \bn, \bk}(x, y; \varepsilon)|
\lesssim  {\mathcal J}_{|\a|+|\b|}^{(1)}(x, y; \varepsilon, R/2) 
+  {\mathcal J}_{|\a|+|\b|}^{(2)}(x, y; \varepsilon, R/2),
\end{align*}
with a constant independent of $\varepsilon$. 
By \eqref{eq:ja3}, the right-hand side satisfies \eqref{eq:11}, as required. 

Now assume that $|\bm|\ge 1$. Using again \eqref{eq:dphi} and \eqref{eq:dnkpsi} we get 
the estimate 
\begin{align*}
|{\mathcal J}_{\bm, \bn, \bk}(x, y; \varepsilon)| \lesssim 
\varepsilon^{-|\bm|}\,\int_{Z(x, y)} 
\CM_\varepsilon(x, y, \hat\bx)
&\ \big(\mu(x,\hat\bx)\big)^{-|\bn|}\,\big(\tilde\mu(y,\hat\bx)\big)^{-|\bk|}\\
&\ f_\infty(x, \hat\bx; R/2)\,  f_\infty(y, \hat\bx; R/2)\, d\hat\bx.
\end{align*}
By Young's inequality again, for all $\bm, \bn, \bk$ such that 
$|\bm|+|\bn|+|\bk|=|\a|+|\b|\ge 1$, 
we have
\begin{align*}
\varepsilon^{-|\bm|}\big(\mu(x, \hat\bx)\big)^{-|\bn|}
&\ \big(\tilde\mu(y, \hat\bx)\big)^{-|\bk|}
\le  
\frac{|\bm|}{|\a|+|\b|}\varepsilon^{-|\a|-|\b|}\\
&  + \frac{|\bn|}{|\a|+|\b|}
\big(\mu(x, \hat\bx)\big)^{-|\a| - |\b|}
+ \frac{|\bk|}{|\a|+|\b|}
\big(\tilde\mu (y, \hat\bx)\big)^{-|\a| - |\b|}.
\end{align*}
Therefore,
\begin{align*}
|{\mathcal J}_{\bm, \bn, \bk}(x, y; \varepsilon)| \lesssim 
&\ \varepsilon^{-|\a| - |\b|}\,\int_{Z(x, y)} 
\CM_\varepsilon(x, y, \hat\bx)
f_\infty(x, \hat\bx; R/2)\,  f_\infty(y, \hat\bx; R/2)\, d\hat\bx\\
&\ + \int_{Z(x, y)} \big(\mu(x,\hat\bx)\big)^{-|\a|-|\b|}
\CM_\varepsilon(x, y, \hat\bx)\,f_\infty(x, \hat\bx; R/2)\,  f_\infty(y, \hat\bx; R/2)\, d\hat\bx\\
&\ + \int_{Z(x, y)} \big(\tilde\mu(y,\hat\bx)\big)^{-|\a| - |\b|}
\CM_\varepsilon(x, y, \hat\bx)\,f_\infty(x, \hat\bx; R/2)\,  f_\infty(y, \hat\bx; R/2)\, d\hat\bx.
\end{align*}
In the second and in the third integral estimate $|\CM_\varepsilon(x, y, \hat\bx)|\lesssim 1$ and use 
the notation \eqref{eq:ja}:
\begin{align*}
|{\mathcal J}_{\bm, \bn, \bk}(x, y; \varepsilon)|
\lesssim &\  {\mathcal J}_{|\a|+|\b|}^{(1)}(x, y; \varepsilon, R/2)\\
&\ + {\mathcal J}_{|\a|+|\b|}^{(2)}(x, y; \varepsilon, R/2)
+  {\mathcal J}_{|\a| + |\b|}^{(3)}(x, y; \varepsilon, R/2), 
\end{align*} 
with a constant independent of $\varepsilon$. 
By \eqref{eq:ja3}, the right-hand side satisfies \eqref{eq:11}. 

Putting the estimates for $|\bm|=0$ and $|\bm|\ge 1$ together, and summing 
over $\bm, \bn, \bk$, we arrive at \eqref{eq:11}, thereby completing the proof. 
\end{proof}

\begin{cor}\label{cor:11tot}
Suppose that for some $\varepsilon>0$ we have $|x|>\varepsilon$, 
$|y|>\varepsilon$ and $|x-y|>\varepsilon$. 
Then for all $\a, \b$ and $|p|\le 1, |r|\le 1,$ 
and all $R>0$ the bound holds:
\begin{align}\label{eq:11tot}
|\p_x^\a \p_y^\b  \g_{r, p}(x,& y)|\notag\\
\lesssim & \big(1 + |x|^{-|\a|-|\b|} + |y|^{-|\a|-|\b|} 
+ h_{|\a|+|\b|+2}(\varepsilon)\big)\,
\big(\|\rho\|_{\plainL1(B(x, R))}\big)^{\frac{1}{2}}
\big(\|\rho\|_{\plainL1(B(y, R))}\big)^{\frac{1}{2}},
%
%
\end{align}
with an implicit constant that is independent of $x, y$ and $\varepsilon$ but may depend on $R$. 
\end{cor}

\begin{proof} 
In order to use Lemma \ref{lem:11} 
we build a partition of unity consisting of extended cut-offs. 
Recall the notation $\SR = \{1, 2, \dots, N\}$. 
Let $\Xi = \{(j, k)\in\SR\times\SR: j < k\}$. For each subset $\U\subset\Xi$ 
introduce the admissible cut-off (see \eqref{eq:canon} for the definition of admissible cut-offs)
\begin{align*}
\phi_{\U}(\bx; \varepsilon) = \prod_{(j, k)\in \U} \z_\varepsilon(x_j-x_k) 
\prod_{(j, k)\in\U^{\rm c}} \t_\varepsilon(x_j-x_k).
\end{align*} 
It is clear that 
\begin{align*}
\sum_{\U\subset{\Xi}} \phi_{\Upsilon}(\bx; \varepsilon) 
= \prod_{(j, k)\in \Xi}\big(\z_\varepsilon(x_j-x_k)+ \t_\varepsilon(x_j-x_k)\big) 
= 1.
\end{align*}
Furthermore, for every cluster $\SfS\subset \SR^*$ define 
\begin{align*}
\tau_{\SfS}(y, \hat\bx; \varepsilon) 
= \prod_{j\in \SfS} \z_\varepsilon(y-x_j) 
\prod_{j\in(\SfS^{\rm c})^*} \t_\varepsilon(y-x_j).
\end{align*}
It is clear that 
\begin{align*}
\sum_{\SfS\subset \SR^*}\tau_{\SfS}(y, \hat\bx; \varepsilon) 
= \prod_{j\in\SR^*}\big(\z_\varepsilon(x_1-x_j)+ \t_\varepsilon(x_1-x_j)\big) = 1.  
\end{align*}
Define
\begin{align*}
\Phi_{\U, \SfS}(x, y, \hat\bx; \varepsilon) 
= \phi_{\U}(x, \hat\bx; \varepsilon) 
\tau_{\SfS}(y, \hat\bx; \varepsilon),\quad 
(x, y)\in \R^3\times \R^3,\, \hat\bx\in \R^{3N-3},
\end{align*}
so that 
\begin{align*}
\sum\limits_{\U\subset\Xi,\ \SfS\subset\SR^*}\Phi_{\U, \SfS}(x, y, \hat\bx; \varepsilon) = 1.
\end{align*}
Each function $\Phi_{\U, \SfS}$ is an extended cut-off function, as 
defined in Subsect. \ref{subsect:extended}. 
Using the definition \eqref{eq:rpphi}, 
the function \eqref{eq:rp} can be represented as 
 \begin{align*}
 \g_{r, p}(x, y) = \sum_{\U\subset\Xi, \ \SfS\subset\SR^*} \g(x, y; \Phi_{\U, \SfS}, \varepsilon).
 \end{align*}
 Applying Lemma \ref{lem:11} to each summand we arrive at \eqref{eq:11tot}. 
\end{proof}

\begin{proof}[Proof of the bound \eqref{eq:der2}] 
Assume that $x\not = 0, y\not = 0, x\not = y$ and that $|l|\ge 1$, $|m|\ge 1$. 
Represent $l = \a+r$, $m = \b+p$ with $|r|=|p|=1$, so that $|l| = |\a|+1, |m| = |\b| +1$.
Furthermore, denote 
\begin{align*}
\varepsilon = \frac{1}{2}\min\{|x|, |y|, |x-y|\},
\end{align*}
so that $|x-y|>\varepsilon$ and $|x|>\varepsilon, |y|>\varepsilon$. 
Thus it follows from \eqref{eq:11tot} that 
\begin{align}\label{eq:11totn}
|\p_x^l \p_y^m  \g(x,& y)|\notag\\
\lesssim & \big(|x|^{2-|l|-|m|} + |y|^{2-|l|-|m|} 
+ h_{|l|+|m|}(\varepsilon)\big)\,
\big(\|\rho\|_{\plainL1(B(x, R))}\big)^{\frac{1}{2}}
\big(\|\rho\|_{\plainL1(B(y, R))}\big)^{\frac{1}{2}}.
%
%
\end{align}
Since for all $a\ge 0$ we have   
\begin{align*}
h_a(\varepsilon)\lesssim &\ 
h_a(|x-y|) + h_a(|x|) + h_a(|y|)\\
\lesssim &\ 1 + |x|^{2-a} + |y|^{2-a} + h_a(|x-y|),
\end{align*} 
the bound \eqref{eq:11totn} implies \eqref{eq:der2}. 
\end{proof}

\begin{proof}[Proof of \eqref{eq:der1}] 
For convenience denote
\begin{align*}
A(x, y; R_1, R_2) = 
\big(\|\rho\|_{\plainL1(B(x, R_1))}\big)^{\frac{1}{2}}
\big(\|\rho\|_{\plainL1(B(y, R_2))}\big)^{\frac{1}{2}}.
\end{align*}
Assume that $|l|\ge 1$. Representing $l = \a + r$, with some $r:|r|=1$,  
and taking $\b = p = 0$,  
%
%
we obtain from \eqref{eq:11tot} that   
\begin{align}\label{eq:xonly}
|\p_x^l \g(x, y)| + &\ |\p_y^l \g(x, y)|\notag\\
\lesssim &\ \big(1+|x|^{1-|l|} + |y|^{1-|l|} 
+ h_{|l|+1}(|x-y|)\big)\, A(x, y; R/2, R/2). 
\end{align}
To "upgrade" this bound to \eqref{eq:der1} we use the 
Fundamental Theorem of Calculus. Let $z\in\R^3$, $z\not = x,$ be such that 
the segment 
\begin{align*}
\{y_s = y+s(z-y),\, s\in [0, 1]\},
\end{align*}
does not contain the point $x$. Then 
\begin{align}\label{eq:calc}
\p_x^l \g(x, y) - \p_x^l \g(x, z) = - \int_0^1 (z-y)\cdot\nabla_y(\p_x^l\g)(x, y_s)\, ds. 
\end{align}
A similar formula holds for the derivative 
$\p_y^l \g(x, y)$, but we omit this part of the argument and complete the proof 
for the derivative $\p_x^l\g(x, y)$ only. 
Since the integrand in \eqref{eq:calc} contains derivatives both w.r.t. $x$ and $y$, 
we can use the bound \eqref{eq:der2} proved previously. 
First we make a convenient choice of $z$. Denote  
\begin{align*}
|x-y| = d,\quad \d = \frac{1}{4}\min\{|x|, |y|\}, 
\quad \d_1 = \min\{1, \d, R/2\}, \quad 
{\sf e} = \frac{x-y}{|x-y|}, 
\end{align*}
so 
\begin{align*}
x = y + |x-y|\,\se = y+d\,\se.
\end{align*}
Take $z = y - \d_1\,\se$, so that 
\begin{align*}
|y-y_s| = s\d_1,\,  |z-y| = \d_1,\ 
|z-x| = \d_1+ d, \quad 
|z|\ge |y| - \d_1 \ge  \frac{3}{4}|y|,
\end{align*}
and 
\begin{align*}
|y_s|\ge |y| - s\d_1 \ge \frac{3}{4}|y|,\quad
|x-y_s|  = d+s\d_1, \quad s\in [0, 1]. 
\end{align*}
Now apply \eqref{eq:der2} with the radius $R/2$ to estimate the integrand in  
\eqref{eq:calc}:
\begin{align}\label{eq:underint}
|(z-y)\cdot\nabla_y(&\p_x^l\g)(x, y_s)|\notag\\ 
\lesssim &\ |z-y|
\big( 1 + 
|y_s|^{1-|l|} + |x|^{1-|l|} + 
 h_{|l|+1}(|x-y_s|)\big) A(x, y_s; R/2, R/2)\notag\\[0.2cm]
 \lesssim &\ \d_1 \big(1 + 
|y|^{1-|l|} + |x|^{1-|l|} + 
 h_{|l|+1}(d+s\d_1)\big) A(x, y; R/2, R).
\end{align} 
An elementary calculation shows that for all $a\ge 0$ the bound holds:
\begin{align*}
\d_1 \int_0^1 h_{a+1}(d+s\d_1) ds\lesssim h_{a}(d).
\end{align*}
Consequently, integrating the bound \eqref{eq:underint} in $s\in [0, 1]$ we obtain from 
\eqref{eq:calc} that 
%
%
\begin{align}\label{eq:wz}
|\p_x^l \g(x, y)|\lesssim |\p_x^l \g(x, z)| 
+ \big(1 + |y|^{1-|l|} + |x|^{1-|l|} + h_{|l|}(|x-y|)\big) A(x, y; R, R).
\end{align}
To estimate $|\p_x^l \g(x, z)|$ we use \eqref{eq:xonly}:
\begin{align*}
 |\p_x^l \g(x, z)| \lesssim  \big(1 + |x|^{1-|l|} + |z|^{1-|l|} 
+ h_{|l|+1}(|x-z|)\big)\, A\big(x, z; R/2, R/2\big).
\end{align*}
As $|z-x| = d+\d_1\ge \d_1$ we can estimate:
\begin{align*}
h_{|l|+1}(|x-z|)\le h_{|l|+1}(\d_1)\lesssim 1 + |x|^{1-|l|} + |y|^{1-|l|}. 
\end{align*} 
Furthermore, as $|y-z| = \d_1\le R/2$, we have $ A\big(x, z; R/2, R/2\big)
\le  A\big(x, y; R/2, R\big)$. 
Consequently, 
\begin{align*}
 |\p_x^l \g(x, z)| \lesssim  \big(1 + |x|^{1-|l|} + |y|^{1-|l|}\big)\, A\big(x, y; R/2, R\big).
\end{align*}
Together with \eqref{eq:wz} this bound entails \eqref{eq:der1}.
\end{proof}

This completes the proof of Theorem \ref{thm:derbounds}.

\section{Appendix}\label{sect:app}

Here we provide an elementary fact concerning extensions of Sobolev spaces. The proof can be found 
in the appendix to \cite{Sobolev2022a}.  

Consider spaces of functions that depend either 
on one variable $x\in\R^d$ or on two variables $(t, x)\in \R^l\times\R^d$. 
Denote $K = \{t\in\R^l: |t| < 1\}$,\ $B = \{x\in\R^d: |x| < 1\}$ and 
$B_0 = B\setminus\{0\}$. 

\begin{prop} \label{prop:remove}
Let the dimension $l$ be arbitrary, let 
$d\ge 2$, $m\ge 1$, and 
%
%
$p\in [d(d-1)^{-1}, \infty]$. Then 
$\plainW{m, p}(B_0) 
= \plainW{m, p}(B)$ and $\plainW{m, p}\big(K\times B_0\big) 
= \plainW{m, p}\big(K\times B\big)$.
\end{prop}  

We use this fact  in Remark \ref{rem:main} to describe smoothness properties of 
the one-particle density matrix $\gamma(x, y)$.

\vskip 0.5cm

\textbf{Acknowledgments.} The authors are grateful to A. Nazarov for his help regarding Proposition 
\ref{prop:reg}, and to D. Edmunds, V. Kozlov, V. Maz'ya, G. Rozenblum, D. Vassiliev 
for their advice on Sobolev spaces.


\begin{thebibliography}{10}
\providecommand{\url}[1]{\texttt{#1}}
\providecommand{\urlprefix}{URL }
\providecommand{\eprint}[2][]{\url{#2}}

\bibitem{BS1970}
M.~S. Birman and M.~Z. Solomyak, \emph{Asymptotics of the spectrum of weakly
  polar integral operators}. Izv. Akad. Nauk SSSR Ser. Mat. \textbf{34}:
  1142--1158, 1970.

\bibitem{Cioslowski2020}
J.~Cioslowski, \emph{Off-diagonal derivative discontinuities in the reduced
  density matrices of electronic systems}. The Journal of Chemical Physics
  \textbf{153(15)}: 154108, 2020.

\bibitem{Cioslowski2021}
J.~Cioslowski, \emph{Reverse engineering in quantum chemistry: How to reveal
  the fifth-order off-diagonal cusp in the one-electron reduced density matrix
  without actually calculating it}. International Journal of Quantum Chemistry
  \textbf{e26651}, 2021.

\bibitem{RDM2000}
A.~Coleman and V.~Yukalov, \emph{Reduced Density Matrices}, \emph{Lecture Notes
  in Chemistry}, vol.~72. Springer-Verlag Berlin Heidelberg, 2000.

\bibitem{Davidson1976}
E.~Davidson, \emph{Reduced Density Matrices in Quantum Chemistry}. Academic
  Press, 1976.

\bibitem{FHOS2002}
S.~Fournais, et~al., \emph{The electron density is smooth away from the
  nuclei}. Comm. Math. Phys. \textbf{228(3)}: 401--415, 2002.

\bibitem{FHOS2004}
S.~Fournais, et~al., \emph{Analyticity of the density of electronic
  wavefunctions}. Ark. Mat. \textbf{42(1)}: 87--106, 2004.

\bibitem{FS2021}
S.~Fournais and T.~Østergaard Sørensen, \emph{Estimates on derivatives of
  Coulombic wave functions and their electron densities}. Journal für die
  reine und angewandte Mathematik (Crelles Journal) \textbf{2021(775)}: 1--38,
  2021.

\bibitem{GilTru2001}
D.~Gilbarg and N.~S. Trudinger, \emph{Elliptic partial differential equations
  of second order}. Classics in Mathematics, Springer-Verlag, Berlin, 2001.
  Reprint of the 1998 edition.

\bibitem{HaKlKoTe2012}
C.~H\"attig, et~al., \emph{Explicitly Correlated Electrons in Molecules}. Chem.
  Rev \textbf{112(1)}: 4–74, 2012.

\bibitem{HearnSob2022}
S.~A. Hearnshaw, P., \emph{Analyticity of the One-Particle Density Matrix}.
  Ann. Henri Poincar\'e \textbf{23}: 707–738, 2022.

\bibitem{Kato1957}
T.~Kato, \emph{On the eigenfunctions of many-particle systems in quantum
  mechanics}. Comm. Pure Appl. Math. \textbf{10}: 151--177, 1957.

\bibitem{LadUra1968}
O.~A. Ladyzhenskaya and N.~N. Ural'tseva, \emph{Linear and quasilinear elliptic
  equations}. Academic Press, New York-London, 1968. Translated from the
  Russian by Scripta Technica, Inc, Translation editor: Leon Ehrenpreis.

\bibitem{LLS2019}
M.~Lewin, E.~H. Lieb, and R.~Seiringer, \emph{Universal Functionals in Density
  Functional Theory}. 2019. \eprint{1912.10424}.

\bibitem{LiebSei2010}
E.~H. Lieb and R.~Seiringer, \emph{The stability of matter in quantum
  mechanics}. Cambridge University Press, Cambridge, 2010.

\bibitem{ReedSimon2}
M.~Reed and B.~Simon, \emph{Methods of modern mathematical physics. {II}.
  {F}ourier analysis, self-adjointness}. Academic Press [Harcourt Brace
  Jovanovich, Publishers], New York-London, 1975.

\bibitem{Sobolev2022}
A.~V. Sobolev, \emph{Eigenvalue asymptotics for the one-particle density
  matrix}. To appear in Duke Math. J. 2022. \eprint{arXiv 2103.11896}.

\bibitem{Sobolev2022a}
A.~V. Sobolev, \emph{Eigenvalue asymptotics for the one-particle kinetic energy
  density operator}. J. Funct. Anal. 2022. \eprint{arXiv 2105.01986}.

\end{thebibliography}
\end{document}